\pdfoutput=1
\documentclass[pdftex,pra,aps,onecolumn,twoside,nofootinbib]{revtex4}

\usepackage{doi}
\usepackage[normalem]{ulem}
\usepackage{braket}
\usepackage{graphicx}
\usepackage{dsfont}
\usepackage[T1]{fontenc}
\usepackage[utf8]{inputenc}
\usepackage{float}
\usepackage{braket}
\usepackage{indentfirst}
\usepackage{enumerate}
\usepackage{url}
\usepackage[dvipsnames]{xcolor}
\usepackage[stable]{footmisc}
\usepackage{fancyhdr}
\usepackage[all]{xy}
\usepackage{graphicx}
\usepackage{youngtab}
\usepackage{mathtools}
\usepackage{enumerate}
\usepackage{textcomp}
\usepackage{newcent}
\usepackage[sc,noBBpl]{mathpazo}
\usepackage{amsmath,amsfonts,amssymb,amsthm}
\usepackage[mathscr]{eucal}
\usepackage{tcolorbox}
\theoremstyle{plain}
\newtheorem{theorem}{Theorem}
\newtheorem{lemma}[theorem]{Lemma}
\newtheorem{proposition}[theorem]{Proposition}
\newtheorem{corollary}[theorem]{Corollary}

\newtheorem{definition}[theorem]{Definition}
\newtheorem{fact}[theorem]{Fact}

\newtheoremstyle{note}{\topsep}{\topsep}{\slshape}{}{\scshape}{}{ }{}
\theoremstyle{note}

\newtheorem{remark}[theorem]{Remark}

\newcommand\tr{\operatorname{Tr}}

\newcommand\Mn{\widetilde{\Pi}_N}

\newcommand\B{\widetilde{B}}
\newcommand{\<}{\langle}
\renewcommand{\>}{\rangle}
\newcommand{\ind}{\operatorname{ind}}

\newcommand\be{\begin{equation}}
\newcommand\ee{\end{equation}}
\newcommand\bea{\begin{array}}
	\newcommand\eea{\end{array}}
\newcommand\ben{\begin{eqnarray}}
\newcommand\een{\end{eqnarray}}
\newcommand\ot{\otimes}
\definecolor{forest}{RGB}{11, 102,35}

\newcommand{\A}{\mathcal{A}_{n}^{t_{n}}(d)}

\newcommand{\Hom}{\operatorname{Hom}}
\newcommand{\Span}{\operatorname{span}}

\newcommand\bei{\begin{itemize}}
	\newcommand\eei{\end{itemize}}
\newcommand\bee{\begin{enumerate}}
	\newcommand\eee{\end{enumerate}}

\DeclarePairedDelimiter\floor{\lfloor}{\rfloor}

\begin{document}
\title{Square-root measurements and degradation of the resource state in port-based teleportation scheme}

\author{Micha{\l} Studzi\'nski$^{1}$, Marek Mozrzymas$^{2}$, Piotr Kopszak$^{2}$}
\affiliation{
	$^1$ Institute of Theoretical Physics and Astrophysics, University of Gda\'nsk, National Quantum Information Centre, 80-952 Gda\'nsk, Poland\\
	$^2$ Institute for Theoretical Physics, University of Wrocław
	50-204 Wrocław, Poland }
\begin{abstract}
	Port-based teleportation (PBT) is a protocol of quantum teleportation in which a receiver does not have to apply correction to the transmitted state.  In this protocol two spatially separated  parties can teleport an unknown quantum state only by exploiting joint measurements on number of shared $d-$dimensional maximally entangled states (resource state) together with a state to be teleported and one way classical communication. In this paper we analyse for the first time the recycling protocol for the deterministic PBT beyond the qubit case. In the recycling protocol the main idea is to re-use the remaining resource state after one or many rounds of PBT for further processes of teleportation. The key property is to learn how much the underlying resource state degrades after every round of the teleportation process. We measure this by evaluating quantum fidelity between respective resource states. To do so we first present analysis of the square-root measurements used by the sender in PBT by exploiting the symmetries of the system. 
In particular, we show how to effectively evaluate their square-roots and composition.
	These findings allow us to present the
	explicit formula for the recycling fidelity involving only group-theoretic parameters describing irreducible representations in the Schur-Weyl duality. For the first time, we also analyse the degradation of the resource state for the optimal PBT scheme and show its degradation for all $d\geq 2$. In the both versions, the qubit case is discussed separately resulting in compact expression for fidelity, depending only on the number of shared entangled pairs. 
\end{abstract}
\maketitle

\section{Introduction}
The first quantum teleportation protocol introduced in~\cite{bennett_teleporting_1993} allows for transfer of an unknown quantum state between two spatially separated parties without necessity of exchanging the physical system and has found a lot of important practical and theoretical implications, for example~\cite{boschi_experimental_1998,gottesman_demonstrating_1999,gross_novel_2007, jozsa_introduction_2005, pirandola_advances_2015, raussendorf_one-way_2001}. The protocol, which is illustrated in Figure~\ref{fig:pbt_jpa}, requires pre-shared entanglement and consists of three stages. The first stage is a joint measurement on the state to be teleported and the sender's part of the shared entangled state. The second step involves communicating the classical outcome of the measurement by a classical channel to the receiver.  Finally, the third step requires correction operation, depending on the classical message, which recovers the transmitted state. The requirement of the unitary correction in the last step is a limiting factor, especially when the receiver has limited resources.

The breakthrough has been made by Ishizaka and Hiroshima in 2008. They introduced a novel port-based teleportation protocol (PBT) which does not require unitary correction~\cite{ishizaka_asymptotic_2008,ishizaka_quantum_2009}.
In this setup, illustrated in Figure~\ref{fig:pbt_jpa}, parties share a large \textit{resource state} consisting of $N$ copies of the maximally $d-$dimensional entangled states $|\psi_d^+\rangle^{\otimes N}$, where each pair $|\psi_d^+\rangle=(1/\sqrt{d})\sum_i|ii\>$ is called \textit{port}.
Alice performs a joint measurement on an unknown state $|\psi_{A_0}\>$, which she wishes to teleport, together with her half of the resource state, and communicates the outcome to Bob. The outcome of the measurement indicates the subsystem where the state has been teleported to. To obtain the teleported state, Bob picks up the right port indicated by Alice’s outcome, and no further correction is needed. We distinguish two types of PBT protocols, \textit{deterministic}, where state $|\psi_{A_0}\>$ is always transmitted to the receiver but imperfectly, and \textit{probabilistic}, where parties have to accept some non-zero probability of the failure in transmission, but when succeed the transmission is perfect.  In the first case we ask about the fidelity of the transmitted state while in the latter we are interested in probability of success (here the fidelity is one). In both cases, the perfect transmission (with unit fidelity or unit probability of success) is possible only with infinite resources, when numbers of shared entangled pairs is infinite. This is due to the celebrated non-programming theorem~\cite{Nielsen1997}. To know how PBT protocols behave with finite amount of resources (ports $N$) and local dimension $d$ we must know how the fidelity of the teleported state, or probability of success depend on the mentioned global parameters describing the protocols.  Such analysis have been done for qubits in~\cite{ishizaka_asymptotic_2008,ishizaka_quantum_2009} using $SU(2)^{\otimes N}$ representation approach, while for higher dimensions the problem has been tackled and solved by tools suggested by non-trivial  extension of the Schur-Weyl duality in papers~\cite{Studzinski2017,StuNJP}, with asymptotic analysis presented in~\cite{majenz2} by considering a dual representation to $\overline{U}\in \mathcal{U}(d)$, where the bar denotes complex conjugation. Both types of PBT we have their optimal versions, where Alice optimises simultaneously the measurements and shared entangled states~\cite{ishizaka_quantum_2009,StuNJP} before she runs the teleportation procedure. This optimising procedure increases the efficiency of the protocols measured in the number of shared maximally entangled pairs. In particular, in deterministic qubit scheme~\cite{ishizaka_quantum_2009} the entanglement fidelity $F$ scales as $1-\mathcal{O}(1/N^2)$ in optimal protocol and as $1-\mathcal{O}(1/N)$ in non-optimal one. In probabilistic qubit version~\cite{ishizaka_quantum_2009} the probability of success $p$ scales as $1-\mathcal{O}(1/N)$ in optimal protocol and as $1-\mathcal{O}(1/\sqrt{N})$ in non-optimal scheme. In every variant we have square improvement when moving to optimal procedure. The very elegant and full analysis of asymptotic performance of PBT scheme in all variants, and an arbitrary dimension of the port $d$ is contained in~\cite{majenz2}. However, increasing $d$ does not change the scaling in $N$ in every version.
\begin{figure}[h]
	\centering
	\includegraphics[width=0.7\columnwidth,keepaspectratio,angle=0]{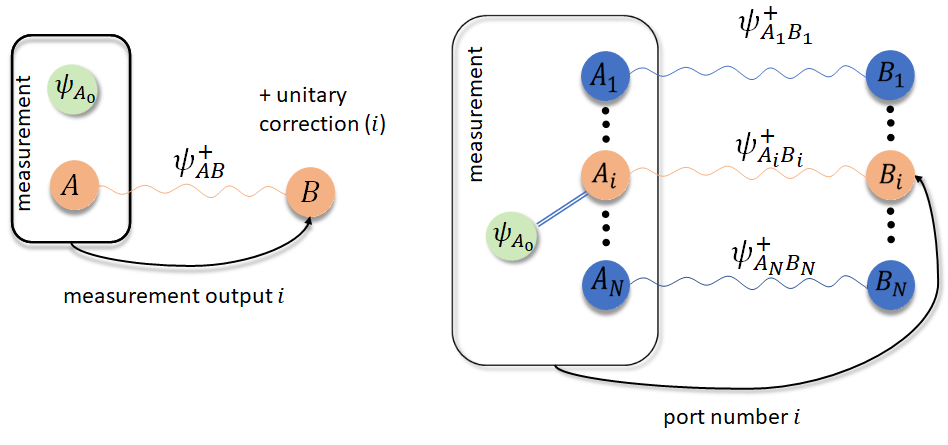}
	\caption{The left panel presents schematic description of the standard quantum teleportation procedure introduced in~\cite{bennett_teleporting_1993}. In this scheme the receiver to recover the transmitted state must apply unitary correction which depends on classical message send by the sender. In the right panel we present schematic description the port-based teleportation introduced in~\cite{ishizaka_asymptotic_2008}. Here, in contrary to the standard teleportation scheme, parties share $N$ maximally entangled pairs (ports), and Bob to recover the transmitted state must just pick up the right port according to classical message send by the sender. No correction is needed here, however due to no-programming theorem~\cite{Nielsen1997} transmission is not perfect, resulting with fidelity of teleportation smaller than 1. The perfect transmission is possible only in asymptotic scenario, when $N\rightarrow \infty$.}
	\label{fig:pbt_jpa}
\end{figure}

The PBT protocols due to the lack of correction in the last step have diverse applications and they are  particularly useful in multi-round quantum information processing settings, where the ordinary teleportation fails. For example, we can use PBT in NISQ protocols as a model for universal processor~\cite{ishizaka_asymptotic_2008,Banchi2020}, position-based cryptography~\cite{beigi_konig}, fundamental limitations on quantum channels discrimination~\cite{limit}, connection between non-locality and complexity~\cite{buhrman_quantum_2016}, and many other important results~\cite{Ebler,PhysRevLett.123.210502,Stroing,sim,PhysRevA.59.156,jeong2020generalization}. All these applications show two-fold importance of further investigations in PBT area. On the one hand, we learn about the fundamental limitations on state transfer by quantum teleportation imposed by the laws of quantum mechanics. On the other hand however, we can exploit PBT for producing many theoretical quantum information processing protocols having an impact on developing the applicative side of the science.

Nevertheless, regardless the variation of the PBT scheme the parties have to exploit substantial number of shared maximally entangled states to obtain satisfactory efficiency. These states can be considered as a resource which has to be produced, stored and possibly costly. This means that one would like to reduce potential costs of preparing PBT by for example using remaining ports after every round of teleportation procedure. To check whether we can re-use remaining ports we have to learn how the resource behaves after joint measurement applied by Alice. Such a possibility would have a great impact on possible practical applications of PBT, since one would get rid of the necessity of preparing the resource state after every teleportation process minimising costs and consumed time. The general idea of such kind is known as \textit{recycling protocol for PBT} $\mathcal{P}_{rec}$ and has been introduced firstly for deterministic scheme in~\cite{strelchuk_generalized_2013}. It is clear that efficiency of such protocol depends on the number of ports $N$, local dimension $d$, and the number of rounds $k$, so we should write $\mathcal{P}_{rec}=\mathcal{P}_{rec}(N,d,k)$.

 For the reader's convenience we present below all steps made by the parties in the recycling scheme (taken from~\cite{strelchuk_generalized_2013}), see also Figure~\ref{fig:rec_jpa}:
\begin{enumerate}
	\item Alice performs a measurement $\widetilde{\Pi}_a^{AA_0}$ with $\sum_{a=1}^N\widetilde{\Pi}_a^{AA_0}=\mathbf{1}_{AA_0}$, obtaining an outcome $1\leq a\leq N$.
	\item Alice sends outcome $a$ to Bob by a classical channel.
	\item Parties apply a transposition (SWAP) between $a-$th and 1st port
	\item Parties do not use port 1 in next rounds of the protocol - they only use remaining $N-1$ ports.
	\item Parties repeat steps 1-4 using remaining ports to complete transmission of $k$ states.
\end{enumerate}
\begin{figure}[h]
	\centering
	\includegraphics[width=0.8\columnwidth,keepaspectratio,angle=0]{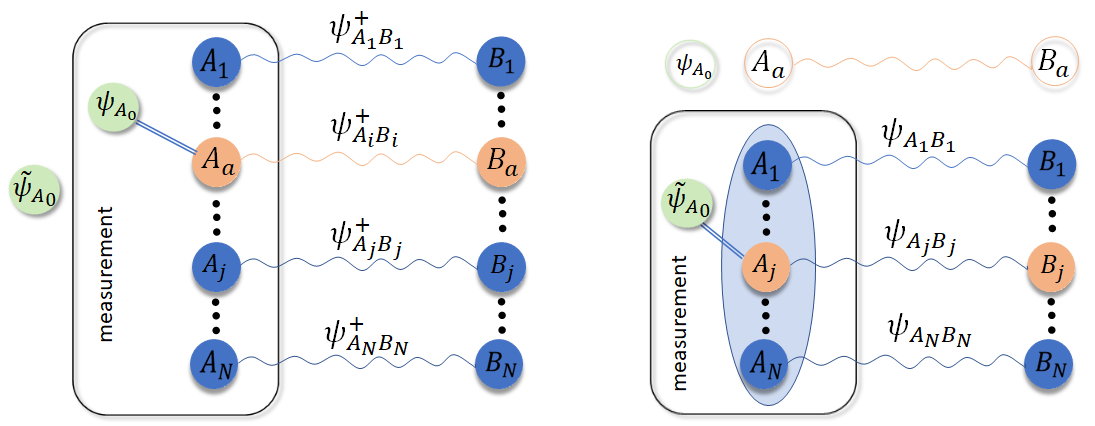}
	\caption{Schematic description of the recycling scheme for teleporting two unknown quantum sates $\psi_{A_0},\widetilde{\psi}_{A_0}$. For the simplicity of the presentation we do not include a classical communication between the parties. On the left we see the usual port-based teleportation procedure, when transmission occurred through $a-$th port. After this parties are left with $N-1$ ports, we do not have port $\psi^+_{A_aB_a}$, since it has been consumed for teleporting state $\psi_{A_0}$. Please notice that after the  measurement in the first round each port is no longer in the form of maximally entangled state, and there are some correlations between all the ports (light blue ellipse). In fact, we have a one large multipartite resource state. Next, in the second round, when parties wish to transmit the state $\widetilde{\psi}_{A_0}$ by using this distorted state. In general one could imagine that the measurement applied by the sender is too destructive, however we show here that for an arbitrary dimension $d$ this distortion is not too big, allowing for further rounds of teleportation.}
	\label{fig:rec_jpa}
\end{figure}
As we will see later the third step is optional for calculating the efficiency of the recycling protocol, and in fact parties could not apply the swap operation or apply any other reordering of the ports without changing quality of the protocol. After completing one round of PBT, the parties are left with $N-1$ ports and it is natural to ask  what is the usefulness of the remaining ports for next teleportation processes. This question has been asked for the first time in~\cite{strelchuk_generalized_2013} together with the description of \textit{the recycling protocol} for PBT. The recycling protocol $\mathcal{P}_{rec}(N,d,k)$ would allow for sequential teleportation of a number of quantum states by exploiting the same resource state in each round. Namely, after each application of PBT the parties do not prepare new $N$ maximally entangled pairs but use the remaining $N-1$ ports of the resource state.

To show that the recycling protocol $\mathcal{P}_{rec}(N,d,k)$ is indeed efficient it is sufficient, as it was explained in~\cite{strelchuk_generalized_2013}, to find  the fidelity between states in the idealised situation, where the state is teleported and the remaining resource state is untouched, and the real state of the resource after application of a joint measurement in PBT. Having this one can check how such fidelity behaves after, let us say $k$ rounds of PBT.  Up to now \textit{only the qubit case}, for non-optimal PBT has been investigated and  there is a lower bound (Theorem 1 in~\cite{strelchuk_generalized_2013}) for the fidelity $F(\mathcal{P}_{rec}(N,2,1))$  of the form:
\begin{equation}
	\label{eq1}
	F(\mathcal{P}_{rec}(N,2,1))\geq 1-\frac{11}{4N}+\mathcal{O}\left(\frac{1}{N^2}\right).
\end{equation}
Next, having a lower bound on fidelity $F(\mathcal{P}_{rec})$ after one round of the recycling protocol, one can establish similar lower bound after $k$ rounds of the protocol (Lemma 2 in~\cite{strelchuk_generalized_2013}):
\begin{equation}
	\label{eq2}
	F(\mathcal{P}_{rec}(N,2,k))\geq 1-\frac{11k}{2N}.
\end{equation}
The above expression states that the error after each round is at most additive in the number of rounds $k$. These results would imply that in every round of teleportation Alice can apply the same type of measurement called \textit{square-root measurement} which is in fact optimal for non-optimal and optimal PBT due to the results in~\cite{leditzky2020optimality}  producing reasonably high efficiency of teleportation when parties re-use the remaining ports.

In this paper we extend the results on recycling protocol for deterministic PBT beyond qubit case and we discuss recycling for the optimal version of PBT. 
In  particular, our contribution is the following:
\begin{enumerate}
	\item To obtain all the results regarding the recycling protocol we have to know the interior structure of the square-root measurements existing in deterministic PBT. In this paper we present substantial analysis of the interior structure of these objects from the point of view of representation theory. We prove several propositions including: their composition law which shows when the considered measurements become projective, we compute matrix elements of the measurements in irreducible blocks, and finally we present how to effectively calculate square-roots from the measurements occurring in PBT. The latter is crucial in computing  recycling fidelity for PBT.
	\item  For \textit{an arbitrary dimension} $d$ of the port we evaluate expressions for $F(\mathcal{P}_{rec}(N,d,1))$ in the case of non-optimal and optimal deterministic PBT in terms of the operators describing the teleportation protocol like Alice's measurements and signal states. In particular, the analysis for the optimal PBT is presented for the first time, even in the qubit case.
	\item In both variants we derive expressions for explicit values of  $F(\mathcal{P}_{rec}(N,d,1))$ depending on group theoretic quantities such as multiplicities and dimensions of irreducible representations of the symmetric groups $S(N)$ and $S(N-1)$ in the Schur-Weyl duality. These results are obtained for an arbitrary port dimension $d$. In particular case, when $d=2$, we present effectively computable expressions depending \textit{only} on the number of ports $N$.
	\item Exploiting already existing results (Lemma 2 in~\cite{strelchuk_generalized_2013}), we present a lower bound on $F(\mathcal{P}_{rec}(N,2,k))$ in terms of group-theoretical quantities for any $d\geq 2$, showing that the error is at most additive in $k$.
	\item We show numerically that there is no clear connection between $F(\mathcal{P}_{rec}(N,d,1))$ and the type of resource states for deterministic PBT. We show that the fidelity between the resource states for non- and optimal deterministic PBT decreases strongly for large $N$, showing qualitative difference between the states. However, this behaviour does not imply the difference between behaviour of  $F(\mathcal{P}_{rec}(N,d,1))$ for discussed schemes. In particular, even that two resource states are very different the values of $F(\mathcal{P}_{rec}(N,d,1))$ do not change significantly.
\end{enumerate}

The structure of the paper is as follows. In Section~\ref{dPBT} we describe in detail deterministic PBT and identify all its symmetries with respect to unitary and symmetric group. In Section~\ref{tools} we introduce the minimal necessary amount of information regarding representation theory of symmetric group $S(n)$ and algebra of partially transposed permutation operators  required for understanding augmentations presented later. In Section~\ref{StructurePOVMs} we analyse the square-root measurements from PBT from the point of view of their 
underlying symmetries. In particular, we evaluate the composition law for them and their square-roots occurring in the recycling protocol.
In Section~\ref{rec_PBT} we formally introduce the recycling protocol for deterministic PBT and present main results of this paper. We start from Theorem~\ref{thm1_rec}, where explicit equation for $F(\mathcal{P}_{rec}(N,d,1))$ is presented as a function of the joint measurement occurring in PBT. 
Next, in Theorem~\ref{Frec_bound}, using Schwarz inequality and properties of the joint measurement we derive an upper bound for $F(\mathcal{P}_{rec}(N,d,1))$, showing that the obtained expression is well defined.
In Theorem~\ref{expplicit} we present explicit expression for $F(\mathcal{P}_{rec}(N,d,1))$ in arbitrary dimension $d$ in terms of group-theoretic parameters like dimensions and multiplicities of irreducible representations in the Schur-Weyl duality. Then, in Lemma~\ref{Fqubits} the reduction to qubit case of the the statement of Theorem~\ref{expplicit} is presented. In the same section we analyse the efficiency of the recycling protocol when Alice optimises over measurements and the resource state simultaneously - see Theorem~\ref{F_rec_optimal} and Lemma~\ref{lem:f_opt}. Lastly, we present a short discussion about  the previously undiscussed connection between the type of the resource state and the efficiency of the recycling protocol. In fact, we show that the resource states for non- and optimal PBT are very different (Lemma~\ref{l:FPBT}) but resulting fidelity in the recycling protocol is almost the same for both of them.  
Our paper contains also appendices where we give detailed proofs of the statements from the main text which require more advanced tools from representation theory. We talk about explicit expressions for $F(\mathcal{P}_{rec}(N,d,1)$ in arbitrary dimension of the port, as well as, its simplification in the qubit case (Appendix~\ref{AppA0}, Appendix~\ref{rec_OdPBT}, Appendix~\ref{app:opbtd2}).

\section{Deterministic Port-based teleportation}
\label{dPBT}
In this section we describe the deterministic version of PBT~\cite{ishizaka_asymptotic_2008,ishizaka_quantum_2009,Studzinski2017,StuNJP} together with the symmetries emerging in the protocol.

	{\bf Deterministic port-based teleportation.} In deterministic PBT parties share a state, called \textit{the resource state}, composed of $N$ copies of $d$-dimensional maximally entangled states, each of them called port. Without loss of generality we assume the following form of shared state:
\be
\label{resource}
|\Psi\>_{AB}=(O_A \otimes \mathbf{1}_B)|\Psi^+\>_{AB}=(O_A \otimes \mathbf{1}_B)|\psi^+\>_{A_1B_1}\otimes |\psi^+\>_{A_2B_2}\otimes \cdots \otimes |\psi^+\>_{A_NB_N},
\ee
where $A=A_1A_2\cdots A_N$, $B=B_1B_2\cdots B_N$, and $O_A$, with normalisation constraint $\tr(O_A^{\dagger}O_A)=d^N$, is a global operation applied by Alice to increase the efficiency of the protocol. In non-optimal PBT $O_A=\mathbf{1}_A$, while for optimal scheme its explicit form in known and discussed in~\cite{ishizaka_quantum_2009,StuNJP}.
Alice to transmit the state of an unknown particle $\psi_C$ performs a joint measurement,  on the state $\psi_{A_0}$ and her half of the resource state. The measurements $\{\widetilde{\Pi}_a^{AA_0}\}_{a=1}^N$ are described here by positive operator valued measure (POVM), so they satisfy the relation $\sum_{a=1}^N\widetilde{\Pi}_a^{AA_0}=\mathbf{1}_{AA_0}$. After the measurement she gets a classical outcome $1\leq a\leq N$ transmitted to Bob by a classical channel. To end the procedure Bob has to just pick-up the right port pointed by the classical message $a$. Denoting by $\Psi_{AB}=|\Psi\>\<\Psi|_{AB}$,   $\Psi_{AB}^+=|\Psi^+\>\<\Psi^+|_{AB}$, and by $\psi_{A_0}=|\psi\>\<\psi|_{A_0}$, we write the teleportation channel $\mathcal{N}$ which has the following form:
\be
\label{ch1}
\begin{split}
	\mathcal{N}\left(\psi_{A_0} \right)&=\sum_{a=1}^N\tr_{A\bar{B}_aA_0}\left[ \sqrt{\widetilde{\Pi}_a^{AA_0}}\left(\Psi_{AB}\ot \psi_{A_0} \right)\sqrt{\widetilde{\Pi}_{a}^{AA_0}}^{\dagger}\right]_{B_a\rightarrow \B}\\
	&=\sum_{a=1}^N\tr_{AA_0}\left[\widetilde{\Pi}_{a}^{AA_0}\left(\left(O_A\ot \mathbf{1}_{B_a} \right)\tr_{\bar{B}_{a}}(\Psi_{AB}^+) \left(O_A^{\dagger}\ot \mathbf{1}_{B_a} \right)\ot \psi_{A_0} \right)\right]_{B_a\rightarrow \B}\\
	&=\sum_{a=1}^N \tr_{AA_0}\left[\widetilde{\Pi}_{a}^{AA_0} \left(\left(O_A\ot \mathbf{1}_{\B} \right)\sigma_{A_a\B} \left(O_A^{\dagger}\ot \mathbf{1}_{\B} \right) \ot \psi_{A_0}\right)\right],
\end{split}
\ee
where by $\tr_{\bar{B}_{a}}$ denotes partial trace over all systems $B$  but $a$. The states $\sigma_{A_a\B}$ are called \textit{signal states} and have the following explicit form
\be
\label{eq:signals}
\sigma_{A_a\B}\equiv \sigma_a=\tr_{\bar{B}_{a}}(\Psi_{AB}^+)=\tr_{\bar{B}_{a}}\left(P^+_{A_1B_1}\ot P^+_{A_2B_2}\ot \cdots \ot P^+_{A_NB_N} \right)_{B_{a}\rightarrow \widetilde{B}}=\frac{1}{d^{N-1}}\mathbf{1}_{\overline{A}_a}\otimes P^+_{A_a\B},
\ee
where $P^+_{A_a\B}$ is projector on maximally entangled state between systems $A_a$ and $\B$. As it was mentioned in deterministic scheme teleportation always succeeds but the teleported state is distorted. To know how well we perform, one can evaluate entanglement fidelity $F(\mathcal{N})$ of teleportation channel $\mathcal{N}$ when teleporting  a subsystem $C$ from a maximally entangled state $P^+_{CD}$, and computing overlap with the state after perfect transmission $P^+_{\B D}$~\cite{ishizaka_asymptotic_2008,ishizaka_quantum_2009,StuNJP}:
\be
F(\mathcal{N})=\tr\left[P^+_{\B D}(\mathcal{N}_{C}\otimes \mathbf{1}_D)(P^+_{CD})\right]=\frac{1}{d^2}\sum_{a=1}^N\tr\left[\left(O_A^{\dagger}\otimes \mathbf{1}_{\B}\right)\widetilde{\Pi}_a^{A\B}\left(O_A\otimes \mathbf{1}_{\B}\right)\sigma_{A_a\B}\right].
\ee
For an arbitrary dimension $d$ the fidelity $F(\mathcal{N})$ has been evaluated explicitly using methods coming from group representation theory~\cite{Studzinski2017,StuNJP,majenz2}.
Due to the recent result presented in \cite{leditzky2020optimality}, we know that \textit{square-root measurements} (SRM) are optimal in both PBT versions, where parties share entangled pairs only, and when Alice optimises over the shared state and measurements. The optimal measurements in the both cases are of the form:
\begin{equation}
	\label{eq:measurements}
	\forall 1\leq a\leq N \qquad \Pi_a^{AA_0}\equiv \Pi_a=\frac{1}{\sqrt{\rho}}\sigma_{A_aA_0}\frac{1}{\sqrt{\rho}},\quad \text{where}\quad \rho=\sum_{a=1}^N\sigma_{A_aA_0}.
\end{equation}
The operator $\rho^{-1}$ is restricted to the support of $\rho$, so to ensure summation of all POVMs to identity $\mathbf{1}_{AA_0}$ on the whole space $(\mathbb{C}^d)^{\otimes N+1}$, we add to every $\Pi_a^{AA_0}$ an excess term
\begin{equation}
	\label{Delta}
	\frac{1}{N}\Delta=\frac{1}{N}\left(\mathbf{1}_{AA_0}-\sum_{a=1}^N\Pi_a^{AA_0}\right),\quad \text{where}\quad \Delta=\mathbf{1}_{AA_0}-\sum_{a=1}^N\Pi_a^{AA_0},
\end{equation}
having for $1\leq a\leq N$ the new operators of the form
\begin{equation}
	\label{mea2}
	\widetilde{\Pi}_a^{AA_0}=\Pi_a^{AA_0}+\frac{1}{N}\Delta.
\end{equation}
As we discuss later (see also~\cite{ishizaka_asymptotic_2008,ishizaka_quantum_2009,Studzinski2017}) this extra term does not change the entanglement fidelity $F(\mathcal{N})$ of the channel $\mathcal{N}$.

	{\bf Symmetries in port-based teleportation} For the further purposes let us focus here a little bit on symmetries occurring in signals and measurements in deterministic PBT.
Now we are ready to identify all symmetries in PBT. First there is a well known observation that a bipartite maximally entangled state is $U\otimes \overline{U}$ invariant, where the bar denotes complex conjugation of an element $U$ of the unitary group $\mathcal{U}(d)$. This implies the following symmetries of all signal states $\sigma_a$:
\begin{equation}
	\label{sym1}
	\begin{split}
		[U^{\otimes N}\otimes \overline{U},\sigma_a]&=0,\quad \forall \ U\in \mathcal{U}(d),\\
		[V(\pi),\sigma_a]&=0,\quad \forall \ \pi\in S(N-1),
	\end{split}
\end{equation}
where $S(N-1)$ is the symmetric group of $N-1$ elements, $\overline{U}$ acts on $B$, and $U^{\otimes N}$ acts on systems $A=A_1\cdots A_N$.
Construction of the signal states $\sigma_{a}$ allows us to identify an additional symmetry which is the covariance with respect to elements from the group $S(N)$:
\begin{equation}
	\begin{split}
		V(\pi)\sigma_aV^{\dagger}(\pi)=\sigma_{\pi(a)},\quad \forall \ \pi\in S(N).
	\end{split}
\end{equation}
In particular, choosing one signal, let us say $\sigma_N$, any other one can be obtained by just implementing an appropriate operator $V(\pi)$, in this case the element from the coset $S(N)/S(N-1)$, elements of
which  are of the form $V[(a,N-1)]$, for $a=1,\ldots,N-1$, where $(a,N-1)$ denotes transposition between respective systems.
The above considerations imply that the operator $\rho$ from~\eqref{eq:measurements} is invariant with respect to elements from $S(N)$ and the following relation for the measurements $\widetilde{\Pi}_a$ from~\eqref{eq:measurements}:
\begin{equation}
	\label{mes_cov}
	V(\pi)\widetilde{\Pi}_aV^{\dagger}(\pi)=\widetilde{\Pi}_{\sigma(a)},\quad \forall \ \pi\in S(N).
\end{equation}
Now, we observe that any bipartite maximally entangled state $P^+_{XY}$ can be viewed as a partially transposed permutation operator $V[(X,Y)]$ between systems $X$ and $Y$:
\begin{equation}
	P^+_{XY}=\frac{1}{d}V^{t_Y}[(X,Y)],\quad t_Y - \text{partial transposition over system $Y$},
\end{equation}
so operator $\rho$ from~\eqref{eq:measurements} reads
\begin{equation}
	\label{eq:rhoV}
	\rho=\frac{1}{d^N}\sum_{a=1}^N\mathbf{1}_{\overline{A}_a\overline{A_0}}\otimes V^{t_{A_0}}[(A_a,A_0)]\equiv\frac{1}{d^N}\sum_{i=1}^NV^{t_{A_0}}[(A_a,A_0)]\equiv \frac{1}{d^N}\sum_{a=1}^NV'[(a,n)],
\end{equation}
where the bar here denotes here all systems but $A_a, A_0$, and in the last equality we renumbered systems according to rule $A_1\mapsto 1,A_2\mapsto 2,\ldots, A_N\mapsto N, A_0\mapsto n=N+1$. By $'$ we denote partial transposition over $n-$th system.  We will exploit this notation later in this paper, making expressions more compact, especially in appendices where we investigate structure of POVMs.
These symmetries, together with observations above allow us to use group theoretic machinery for the algebra of partially transposed permutation operators~\cite{Moz1,MozJPA} together with the Schur-Weyl duality~\cite{FultonSchur}. We discuss this connection on a deeper level later in this paper.

\section{Symmetric group and algebra of  partially transposed permutation operators}
\label{tools}
For self-consistence of the paper and clarity of the further analysis we briefly remind here basic elements of representation theory of the symmetric group and the algebra of partially transposed permutation operators.

	{\bf Representations of symmetric group $S(n)$} Let us start form considering a permutational representation $V$ of the symmetric group $S(n)$, where $n=N+1$, in the space $\mathcal{H\equiv (\mathbb{C}}^{d})^{\otimes n}$ defined in the following way
\begin{definition}
	\label{repV}
	$V:S(n)$ $\rightarrow \Hom(\mathcal{(\mathbb{C}}^{d})^{\otimes n})$ and
	\be
	\forall \pi \in S(n)\qquad V(\pi ).|e_{i_{1}}\>\otimes |e_{i_{2}}\>\otimes
	\cdots \otimes |e_{i_{n}}\>=|e_{i_{\pi ^{-1}(1)}}\>\otimes |e_{i_{\pi
			^{-1}(2)}}\>\otimes \cdots \otimes |e_{i_{\pi ^{-1}(n)}}\>,
	\ee
	where $d\in \mathbb{N}$ and $\{|e_{i}\>\}_{i=1}^{d}$ is an orthonormal basis of the space $\mathcal{\mathbb{C}}^{d}.$
\end{definition}
Since the representation $V(S(n))$ (or $V_d(S(n))$ to underline the space dimension) is defined in a given basis of the space $\mathbb{C}^d$, it is a matrix representation, and operators $V(\pi)$ just permute basis vectors according to the given permutation $\pi$. The representation $V(S(n))$ extends in a natural way to the
representation of the group algebra $\mathbb{C}[S(n)]$ and in this way we get the algebra
\be
\label{CSn}
\mathcal{A}_{n}(d):= \Span_{\mathbb{C}}\{V(\sigma ):\sigma \in S(n)\}\subset \Hom(\mathcal{(\mathbb{C}}^{d})^{\otimes n})
\ee
of operators representing the elements of the group algebra $\mathbb{C}[S(n)]$.
Note that the algebra $A_{n}(d)$ contains a natural subalgebra
\be
\label{An-1}
\mathcal{A}_{n-1}(d):= \Span_{\mathbb{C}}\{V(\sigma _{n-1}):\sigma _{n-1}\in S(n-1)\}.
\ee
To learn about irreps of the symmetric group $S(n)$ we need to introduce a notion of \textit{partition}. A partition $\alpha$ of a natural number $n$, which we denote as $\alpha \vdash n$, is a sequence of positive numbers $\alpha=(\alpha_1,\alpha_2,\ldots,\alpha_r)$, such that
\be
\alpha_1\geq \alpha_2\geq \cdots \geq \alpha_r,\qquad \sum_{i=1}^r\alpha_i=n.
\ee
The above fact can be represented graphically. Namely, every partition can be visualised as a \textit{Young frame} - a collection of boxes arranged in left-justified rows (see the panel A of Figure~\ref{fig:Yng_diag}).
\begin{figure}[h]
	\centering
	\includegraphics[width=0.8\columnwidth,keepaspectratio,angle=0]{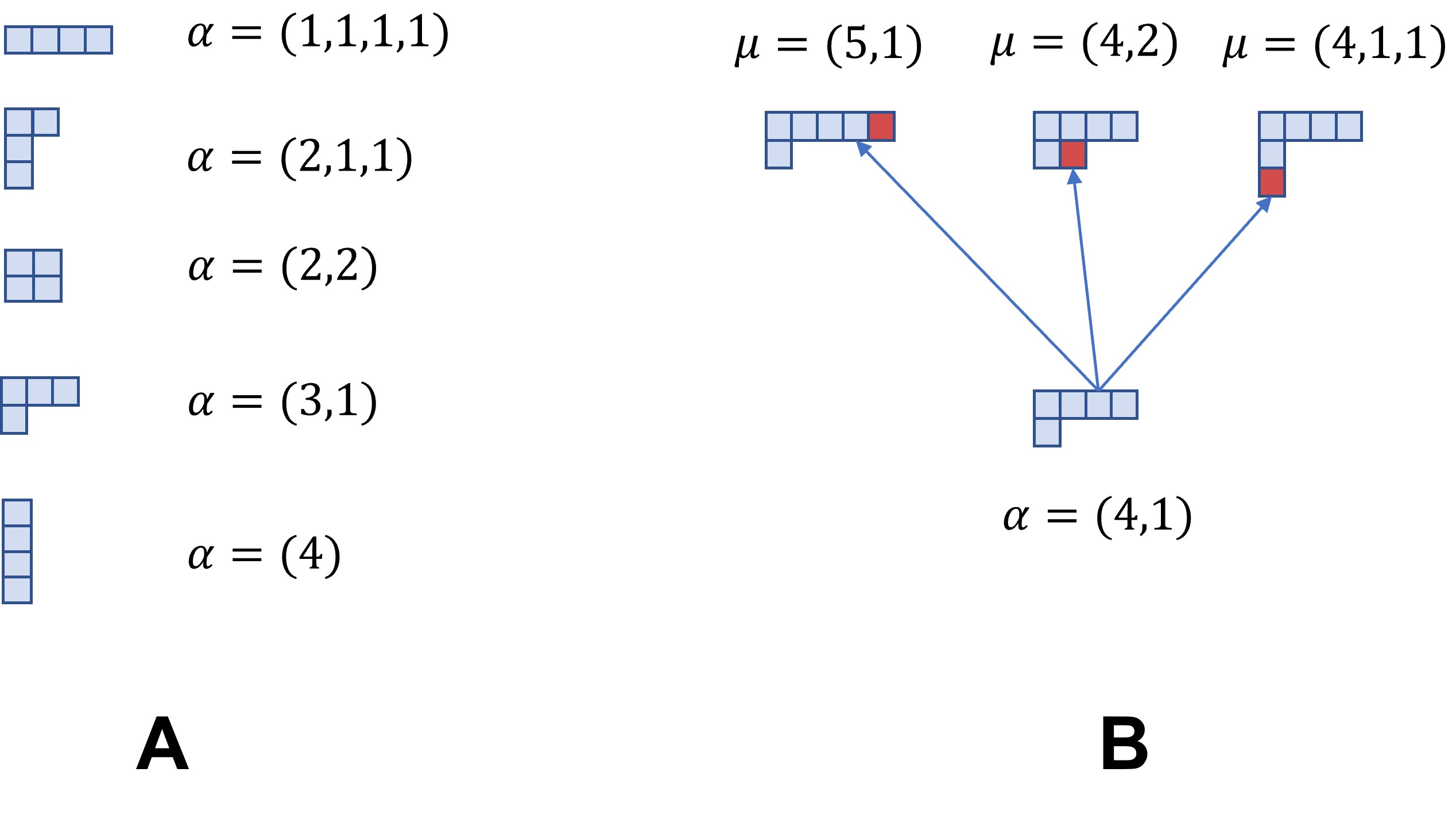}
	\caption{The panel A presents five possible Young frames for $n=4$, which also corresponds to all possible abstract irreducible representations of $S(4)$. Considering representation space $(\mathbb{C}^d)^{\otimes 4}$ there appear only irreps for which height of corresponding Young frames is no larger than $d$. For example, considering qubits ($d=2$) we have only three frames: $(4),(3,1),(2,2)$. The panel B presents possible Young frames $\mu \vdash 6$, which can be obtained from a frame $\alpha=(4,1)$ by adding a single box, depicted here in red. In this particular case, by writing $\mu \in \alpha$, we take $\mu$ represented only by these three frames. In the same manner we define subtracting of a box from a Young frame.}
	\label{fig:Yng_diag}
\end{figure}
For a fixed number $n$, the number of Young frames determines the number of nonequivalent irreps of $S(n)$ in an abstract decomposition.   However, working in the representation space $\mathcal{H\equiv (\mathbb{C}}^{d})^{\otimes n}$, in every decomposition of $S(n)$ into irreps we take Young frames $\alpha$ whose height $h(\alpha)$ is at most $d$. Further, by $\hat{S}(n)$ we denote set of all irreps of  the group $S(n)$.

Now, suppose we have $\alpha \vdash n-1$ and $\mu \vdash n$. Writing $\mu \in \alpha$ we consider such Young frames $\mu$ which can be obtained from $\alpha$ by adding a single box (see the panel B of Figure~\ref{fig:Yng_diag}). Similarly, writing $\alpha \in \mu$ we consider such Young  frames $\alpha$, which can be obtained from $\mu$ by removing a single box.
For further purposes let us define also the following set of irreps of $S(n)$
\begin{equation}
	\label{Theta}
	\Theta:=\left\{\theta \vdash n \ | \ \theta \in \alpha \vdash n-1 \ \text{with} \ h(\alpha)=d \ \text{and} \ h(\theta)=d+1\right\}.
\end{equation}
When one considers irreps of $S(n-1)$ for which $h(\alpha)<d$, then $\Theta$ is a empty set.
Notice that for a given Young frame $\alpha$ with $h(\alpha)=d$ there is only one $\theta$ with $h(\theta)=d+1$.

Finally, having introduced all necessary notation we recall here the celebrated Schur-Weyl duality~\cite{FultonSchur}, which states that the diagonal action of the general linear group $GL_d(\mathbb{C})$ of invertible complex matrices and of the symmetric group on $(\mathbb{C}^d)^{\otimes n}$ commute:
\be
V(\sigma)(X\otimes \cdots \otimes X)=(X\otimes \cdots \otimes X)V(\sigma),
\ee
where $\sigma \in S(n)$ and $X\in GL_d(\mathbb{C})$. Due to the above relation we have the following:
\begin{theorem}
	\label{SW}
	The tensor product space $(\mathbb{C}^d)^{\otimes n}$ can be decomposed as
	\be
	(\mathbb{C}^d)^{\otimes n}=\bigoplus_{\substack{\alpha \vdash n \\ h(\alpha)\leq d}} \mathcal{U}_{\alpha}\otimes \mathcal{S}_{\alpha},
	\ee
	where the symmetric group $S(n)$ acts on the space $\mathcal{S}_{\alpha}$ and the general linear group $GL_d(\mathbb{C})$ acts on the space $\mathcal{U}_{\alpha}$, labelled by the same partitions.
\end{theorem}
From the decomposition given in Theorem~\ref{SW} we deduce that for a given irrep $\alpha$ of $S(n)$, the space $\mathcal{U}_{\alpha}$ is multiplicity space of dimension $m_{\alpha}$ (multiplicity of irrep $\alpha$), while the space $\mathcal{S}_{\alpha}$ is representation space of dimension $d_{\alpha}$ (dimension of irrep $\alpha$).
Finally with every subspace $\mathcal{U}_{\alpha}\otimes \mathcal{S}_{\alpha}$ we associate \textit{Young projector}:
\be
\label{Yng_proj}
P_{\alpha}=\frac{d_{\alpha}}{n!}\sum_{\sigma \in S(n)}\chi^{\alpha}(\sigma^{-1})V(\sigma),\quad \text{with}\quad \tr P_{\alpha}=m_{\alpha}d_{\alpha},
\ee
where $\chi^{\alpha}(\sigma^{-1})$ is the character associated with the irrep indexed by $\alpha$. The symbols $m_{\alpha}, d_{\alpha}$ denote the multiplicity and dimension of an irrep $\alpha$ in the Schur-Weyl dulaity in Theorem~\eqref{SW}. Further, whenever we mean a matrix representation of an irrep of $\sigma \in S(n)$ indexed by a frame $\alpha$ we write $\psi^{\alpha}(\sigma)$ or $\varphi^{\alpha}(\sigma)$.

	{\bf Algebra of partially transposed permutation operators} Having definition of the group algebra $\mathbb{C}[S(n)]$ in equation~\eqref{CSn}, we can naturally introduce the algebra of partially transposed operators with respect to last subsystem $\A$ in the following way
\begin{definition}
	\label{def_A}
	For $\mathcal{A}_{n}(d):= \Span_{\mathbb{C}}\{V(\sigma ):\sigma \in S(n)\}$ we define a new complex algebra
	\be
	\mathcal{A}_{n}^{t_{n}}(d):= \Span_{\mathbb{C}}\{V^{t_{n}}(\sigma ):\sigma \in S(n)\}\subset \Hom(\mathcal{(\mathbb{C}}^{d})^{\otimes n}),
	\ee
	where the symbol $t_{n}$ denotes the partial transposition with respect to the last subsystem
	in the space $\Hom((\mathbb{C}^{d})^{\otimes n})$. The elements $V^{t_{n}}(\sigma ):\sigma \in S(n)$ will
	be called natural generators of the algebra $\mathcal{A}_{n}^{t_{n}}(d)$. Later for the simplicity of the presentation we use  symbol $'$ for partial transposition $t_n$, and $V'$ for transposed permutation operator $V^{t_n}[(n-1,n)]$ between systems $n-1$ and $n$.
\end{definition}
Please notice that from the above definition and expression~\eqref{An-1} it directly follows that $\mathcal{A}_{n-1}(d)\subset
	\mathcal{A}_{n}^{t_{n}}(d)$. It means the algebra $\A$ contains operators representing the subgroup $S(n-1)\subset S(n)$, which are invariant with respect to partial transposition $t_n$.
By the definition the algebra $\A$, which is in fact the a matrix
algebra, acts naturally in the space $\mathcal{H}=(\mathbb{C}^{d})^{\otimes n}.$ From papers~\cite{Moz1,Studzinski2017} we know that the algebra $\A$ is a direct sum of two ideals
\begin{equation}
	\label{A_decomp}
	\A=\mathcal{M}\oplus \mathcal{S}=F\A F\oplus
	(id_{\A}-F)\A(id_{\A}-F),
\end{equation}
where the idempotent $F=\sum_{\alpha \vdash n-2}\sum_{\mu \in \alpha}F_{\mu }(\alpha )$ is the identity on the ideal $\mathcal{M}$, i.e. $F=id_{\mathcal{M}}$. The operators $F_{\mu }(\alpha )$ are projectors on irreps of $\A$ contained in the ideal  $\mathcal{M}$.
The
ideals $\mathcal{M}$ and $\mathcal{S}$ also act in the space $\mathcal{H}=(\mathbb{C}^{d})^{\otimes n}$. The idempotents $F$
and $id_{\A}-F$ satisfy the relation
\begin{equation}
	F+(id_{\A}-F)=id_{\A},\qquad
	F(id_{\A}-F)=0=(id_{\A}-F)F.
\end{equation}%
These properties of the projectors $F$ and $(id_{\A}-F)$ imply,
that the carrier space $\mathcal{H}$ of the algebra $\A$
splits into a direct sum of two orthogonal subspaces
\begin{equation}
	\mathcal{H=}F(\mathcal{H})\oplus (id_{\A}-F)(\mathcal{%
		H})\equiv \mathcal{H}_{\mathcal{M}}\oplus \mathcal{H}_{\mathcal{S}}.
\end{equation}%
and we have
\begin{equation}
	\forall m\in \mathcal{M}\quad m\mathcal{H}_{S}=0,
\end{equation}%
i.e. all elements of the ideal $\mathcal{M}$ act trivially on the subspace $\mathcal{H}%
	_{\mathcal{S}}$, so we have
\begin{equation}
	\forall m\in \mathcal{M}\quad \tr_{\mathcal{H}}(m)=\tr_{\mathcal{H}_{\mathcal{M}}}(m).
\end{equation}

\section{Structure of Square-root measurements in port-based teleportation}
\label{StructurePOVMs}
In this section we investigate the internal structure of POVMs $\{\Pi_a\}_{a=1}^N$ given in~\eqref{eq:measurements} and used by Alice in deterministic PBT scheme. In particular, our main goal here is to calculate the overlap of the signal states $\{\sigma_a\}_{a=1}^N$ with square-roots of POVMs $\{\sqrt{\Pi_a}\}_{a=1}^N$. As a byproduct we also 
prove a composition law for the square-root measurements (Proposition~\ref{p12}) giving conditions for their projectivity.

Let us start from general considerations and for the time being let us drop the extra term $\Delta$ from~\eqref{Delta} in every $\Pi_a$ and write
\begin{equation}
	\sum_{a=1}^N\Pi_a=\frac{1}{\sqrt{\rho}}\sum_{a=1}^N\sigma_a\frac{1}{\sqrt{\rho}}=\frac{1}{\sqrt{\rho}}\rho \frac{1}{\sqrt{\rho}}=id_{\operatorname{supp}(\rho)},
\end{equation}
where $\operatorname{supp}(\rho)$ denotes the support of the operator $\rho$. On the other hand from expression~\eqref{rho_spectral} and interpretation of the projectors $F_{\nu}(\alpha)$ introduced in Section~\ref{tools} one can conclude that $id_{\operatorname{supp}(\rho)}=id_{\mathcal{M}}$, where $\mathcal{M}$ denotes ideal in the decomposition of the algebra $\A$ in~\eqref{A_decomp}. Indeed, as we explained in the proof of Theorem~\ref{Frec_bound}, for computing of the mentioned overlap, we do not have to take into account $\Delta$, since $\operatorname{supp}(\Delta)\perp \operatorname{supp}(\sigma_a)$ for $1\leq a\leq n-1$.

It appears that further properties of the operators $\{\Pi_a\}_{a=1}^{n-1}$
depend on the relation between the numbers $d$ and $n$, i.e. between dimension of the port $d$ and total number of systems in $\mathcal{H}=(\mathbb{C}^d)^{\otimes n}$. It follows from~\cite{Moz1,MozJPA} that if $d\geq n-1$, then the irrep $M_{f}^{\alpha }$ in reduced basis $f$ of the algebra $
	A_{n}^{t_{n}}(d)$ is the full induced representation $\Phi ^{\alpha
	}=\ind_{S(n-2)}^{S(n-1)}(\alpha )$ of the subalgebra $V_{d}(S(n-1))$, i.e. we
have
\be
\label{full}
M_{f}^{\alpha }=\bigoplus _{\nu \in \Phi ^{\alpha }}\psi ^{\nu },
\ee
as a representation of $S(n-1)$, but if $d<n-1$ then we have
\be
\label{notfull}
M_{f}^{\alpha }=\bigoplus _{\nu \in \Phi ^{\alpha },\nu \neq \theta }\psi ^{\nu
},
\ee
where $\psi ^{\theta }$ is the irrep of $S(n-1)$ which does not occur in the decomposition. It takes place when height $h(\cdot)$ of a Young frame $\alpha$ satisfies $h(\alpha)=d$.

First we find an expression for the matrix elements of $M_f^{\alpha}[\Pi_a]$ of a given POVM $\Pi_a$ in the irrep $M_f^{\alpha}$ in the  reduced basis $f\equiv\{f_{j_{\nu} }^{\nu }:h(\nu )\leq d,\quad j_{\nu} =1,\ldots,d_{\nu }\}$ of the ideal $\Phi^{\alpha }$:
\begin{proposition}
	\label{p11}
	The matrix elements of POVM $\Pi_a$, where $1\leq a\leq n-1$, in the irrep $M_f^{\alpha}$ in reduced basis $f$ of the algebra $\A$ are the following:
	\be
	M_{f}^{\alpha }[\Pi _{a}]_{j_{\xi _{\omega }}j_{\zeta _{\nu }}}^{\xi
	_{\omega }\zeta _{\nu }}=\frac{1}{n-1}\frac{\sqrt{d_{\omega }d_{\nu }}}{%
	d_{\alpha }}\sum_{k_{\alpha }}\psi _{R}^{\omega }[(a,n-1)]_{j_{\xi _{\omega
			}}k_{\alpha }}^{\xi _{\omega }\alpha }\psi _{R}^{\nu }[(a,n-1)]_{k_{\alpha
		}j_{\zeta _{\nu }}}^{\alpha \zeta _{\nu }},
	\ee
	where $\omega,\nu \neq \theta$ if $h(\alpha)=d$.
\end{proposition}

\begin{proof}
	In the irrep $\Phi
		^{\alpha }$ of the algebra $\mathcal{A}_{n}^{t_{n}}(d)$ in PRIR basis $M_f^{\alpha}$ (see Appendix~\ref{Prir} for a short summary) the  matrix form for the operators $
		V^{t_{n}}[(a,n)]$ is given through expression~\ref{blee1}.
	Next we know by Lemma 35 in~\cite{MozJPA} that in the irrep $\Phi ^{\alpha }$ of the algebra $%
		\mathcal{A}_{n}^{t_{n}}(d)$ in PRIR basis $M_f^{\alpha}$ the operator $\rho$ from~\eqref{rho_spectral} is diagonal
	\be
	M_{f}^{\alpha }[\rho ]_{j_{\xi _{\omega }}j_{\zeta _{\nu }}}^{\xi _{\omega
	}\zeta _{\nu }}=\delta ^{\omega \nu }\delta ^{\xi _{\omega }\zeta _{\nu
		}}\delta _{j_{\xi _{\omega }}j_{\zeta _{\nu }}}\lambda _{\nu }(\alpha ),
	\ee
	where the numbers $\lambda _{\nu }(\alpha )$ are given in~\eqref{llambda}. Therefore we have
	\be
	M_{f}^{\alpha }\left[ \frac{1}{\sqrt{\rho }}\right] _{j_{\xi _{\omega }}j_{\zeta _{\nu
	}}}^{\xi _{\omega }\zeta _{\nu }}=\delta ^{\omega \nu }\delta ^{\xi _{\omega
		}\zeta _{\nu }}\delta _{j_{\xi _{\omega }}j_{\zeta _{\nu }}}\frac{1}{\sqrt{%
			\lambda _{\nu }(\alpha )}}
	\ee
	and further
	\be
	M_{f}^{\alpha }[\Pi _{a}]_{j_{\xi _{\omega }}j_{\zeta _{\nu }}}^{\xi
	_{\omega }\zeta _{\nu }}=\frac{1}{n-1}\frac{\sqrt{d_{\omega }d_{\nu }}}{%
	d_{\alpha }}\sum_{k_{\alpha }}\psi _{R}^{\omega }[(a,n-1)]_{j_{\xi _{\omega
			}}k_{\alpha }}^{\xi _{\omega }\alpha }\psi _{R}^{\nu }[(a,n-1)]_{k_{\alpha
		}j_{\zeta _{\nu }}}^{\alpha \zeta _{\nu }}.
	\ee
	This finishes the proof.
\end{proof}

\begin{proposition}
	\label{p12}
	For any PRIR representation $M_{f}^{\alpha}$ and POVM operators $\{\Pi_a\}_{a=1}^{n-1}$, we have
	\begin{equation}
		\label{eq:p12}
		\forall 1\leq a\leq n-1 \quad M_{f}^{\alpha }[\Pi _{a}]M_{f}^{\alpha }[\Pi _{a}]=\left(1-\frac{d_{\theta }}{%
				(n-1)d_{\alpha }}\right)M_{f}^{\alpha }[\Pi _{a}].
	\end{equation}%
	If $h(\alpha )<d$ then $d_{\theta }=0$, and $M_{f}^{\alpha
			}[\Pi _{a}]$ is  a projector. If $h(\theta )=d$, then $%
		d_{\theta }\neq 0$ and $M_{f}^{\alpha }[\Pi _{a}]$ is a pseudo-projector.
\end{proposition}

\begin{proof}
	For the proof we use expression for the matrix elements of $\Pi_a$ presented in Proposition~\ref{p11}. Let us calculate the composition in~\eqref{eq:p12} in PRIR indices:
	\begin{align}
		\left[M_{f}^{\alpha }[\Pi _{a}]M_{f}^{\alpha }[\Pi _{a}]\right]^{\xi_{\omega}\xi_{\nu}}_{j_{\xi_{\omega}}j_{\xi_{\nu}}} & =\frac{1}{(n-1)^2}\sum_{\substack{\rho \in \Phi^{\alpha}                    \\ \rho\neq \theta}}\sum_{k_{\alpha},l_{\alpha}}\frac{\sqrt{d_{\omega}d_{\nu}}d_{\rho}}{d_{\alpha}^2} \psi_R^{\omega}[(a,n-1)]^{\xi_{\omega} \ \alpha}_{j_{\xi_{\omega}}k_{\alpha}} \psi_R^{\rho}[(a,n-1),(a,n-1)]^{\alpha \ \alpha}_{k_{\alpha}l_{\alpha}}\times \\
		                                                                                                                        & \times \psi_R^{\nu}[(a,n-1)]^{\alpha \ \xi_{\nu}}_{l_{\alpha}j_{\xi_{\nu}}} \\
		                                                                                                                        & =\frac{1}{(n-1)^2}\sum_{\substack{\rho \in \Phi^{\alpha}                    \\ \rho\neq \theta}}d_{\rho}\sum_{k_{\alpha},l_{\alpha}}\delta_{k_{\alpha}l_{\alpha}}\frac{\sqrt{d_{\omega}d_{\nu}}}{d_{\alpha}^2} \psi_R^{\omega}[(a,n-1)]^{\xi_{\omega} \ \alpha}_{j_{\xi_{\omega}}k_{\alpha}}\psi^{\nu}_R[(a,n-1)]^{\alpha \ \xi_{\omega}}_{l_{\alpha}j_{\xi_{\omega}}}.
	\end{align}
	Observing that
	\be
	\sum_{\substack{\rho \in \Phi^{\alpha}\\ \rho\neq \theta}}d_{\rho}=(n-1)d_{\alpha}-d_{\theta},
	\ee
	we have
	\begin{align}
		\left[M_{f}^{\alpha }[\Pi _{a}]M_{f}^{\alpha }[\Pi _{a}]\right]^{\xi_{\omega}\xi_{\nu}}_{j_{\xi_{\omega}}j_{\xi_{\nu}}} & =\frac{1}{(n-1)^2}\frac{\sqrt{d_{\omega}d_{\nu}}}{d_{\alpha}^2}((n-1)d_{\alpha}-d_{\theta})\sum_{k_{\alpha}} \psi_R^{\omega}[(a,n-1)]^{\xi_{\omega} \ \alpha}_{j_{\xi_{\omega}}k_{\alpha}}\psi_R^{\nu}[(a,n-1)]^{\alpha \ \xi_{\nu}}_{k_{\alpha}j_{\xi_{\nu}}} \\
		                                                                                                                        & =\frac{(n-1)d_{\alpha}-d_{\theta}}{(n-1)d_{\alpha}}M_f^{\alpha}[\Pi_a]^{\xi_{\omega}\xi_{\nu}}_{j_{\xi_{\omega}}j_{\xi_{\nu}}}=\left(1-\frac{d_{\theta}}{(n-1)d_{\alpha}}\right)M_f^{\alpha}[\Pi_a]^{\xi_{\omega}\xi_{\nu}}_{j_{\xi_{\omega}}j_{\xi_{\nu}}},
	\end{align}
	where in the second equality we use direct expression for $M_f^{\alpha}[\Pi_a]^{\xi_{\omega}\xi_{\nu}}_{j_{\xi_{\omega}}j_{\xi_{\nu}}}$ from Proposition~\ref{p11}. We see that whenever $d_{\theta}\neq 0$ the POVMs $\{\Pi_a\}_{a=1}^N$ are pseudo-projectors with the factor $1-\frac{d_{\theta}}{(n-1)d_{\alpha}}$, this is always the case when $d\leq n-1$. Finally when $h(\alpha)<d$, which is always the case when $d\geq n-1$, then $d_{\theta}=0$, since there are no irreps to remove, and the above equation reduces to
	\begin{equation}
		\left[M_{f}^{\alpha }[\Pi _{a}]M_{f}^{\alpha }[\Pi _{a}]\right]^{\xi_{\omega}\xi_{\nu}}_{j_{\xi_{\omega}}j_{\xi_{\nu}}}=M_f^{\alpha}[\Pi_a]^{\xi_{\omega}\xi_{\nu}}_{j_{\xi_{\omega}}j_{\xi_{\nu}}}
	\end{equation}
	showing that POVMs $\{\Pi_a\}_{a=1}^N$ are projectors in this regime.
\end{proof}

Having the above, we are in position to compute the square root from a given POVM $\Pi_a$:
\begin{proposition}
	\label{p13}
	For any PRIR representation $M_{f}^{\alpha }$ in reduced basis $f$ and any POVM operator $\Pi_a$, we have%
	\begin{equation}
		\label{eq:sqrt}
		M_{f}^{\alpha }[\sqrt{\Pi _{a}}]_{j_{\xi _{\omega }}j_{\zeta _{\nu }}}^{\xi
		_{\omega }\zeta _{\nu }}=\frac{1}{\sqrt{(n-1)d_{\alpha }-d_{\theta }}}\frac{%
		\sqrt{d_{\omega }d_{\nu }}}{\sqrt{(n-1)d_{\alpha }}}\sum_{k_{\alpha }}\psi
		_{R}^{\omega }[(a,n-1)]_{j_{\xi _{\omega }}k_{\alpha }}^{\xi _{\omega }\alpha
		}\psi _{R}^{\nu }[(a,n-1)]_{k_{\alpha }j_{\zeta _{\nu }}}^{\alpha \zeta _{\nu
		}}.
	\end{equation}%
\end{proposition}

\begin{proof}
	For the proof it is enough to deduce from Proposition~\ref{p12} that for $1\leq a\leq n-1$ one has
	\begin{equation}
		M_f^{\alpha}[\sqrt{\Pi_a}]=\frac{1}{\sqrt{1-\frac{d_{\theta}}{(n-1)d_{\alpha}}}}M_f^{\alpha}[\Pi_a].
	\end{equation}
	Writing the above in PRIR indices and using the statement of Proposition~\ref{p11} we obtain expression~\eqref{eq:sqrt}.
\end{proof}

Using this we get

\begin{proposition}
	For any PRIR representation $M_{f}^{\alpha }$ we have%
	\begin{equation}
		\label{neqzero}
		\tr\left(M_{f}^{\alpha }[\sqrt{\Pi _{a}}]M_{f}^{\alpha }[V^{\prime }[(a,n)]]\right)=\frac{%
		1}{\sqrt{(n-1)d_{\alpha }-d_{\theta }}}\frac{\sqrt{d_{\alpha }}}{\sqrt{(n-1)}%
		}\frac{1}{m_{\alpha }}\left(\sum_{\substack{\nu \in \Phi^{\alpha}\\ \nu\neq \theta}}\sqrt{m_{\nu }d_{\nu }}\right)^{2}.
	\end{equation}%
	In the case $h(\alpha )<d$, when $d_{\theta }=0$, we have
	\begin{equation}
		\label{zero}
		\tr\left(M_{f}^{\alpha }[\sqrt{\Pi _{a}}]M_{f}^{\alpha }[V^{\prime }[(a,n)]]\right)=\frac{%
			1}{(n-1)m_{\alpha }}\left(\sum_{\nu \in \Phi^{\alpha}  }\sqrt{m_{\nu }d_{\nu }}%
		\right)^{2}.
	\end{equation}
\end{proposition}

\begin{proof}
	First we prove expression~\eqref{zero}, when $h(\alpha)<d$. It means that in this particular case one has $d_{\theta}=0$ and the irrep $M_f^{\alpha}$ is the full induced representation at it is described in~\eqref{full}. Taking form of $M_f^{\alpha}[V'[(a,n)]]^{\xi_{\nu} \ \xi_{\rho}}_{j_{\xi_{\nu}}j_{\xi_{\rho}}}$ from Proposition~\ref{BP20} and form of $M_f^{\alpha}[\sqrt{\Pi_a}]^{\xi_{\omega} \ \xi_{\nu}}_{j_{\xi_{\omega}}j_{\xi_{\nu}}}$ from Proposition~\ref{p13}, we write:
	\begin{align}
		 & \sum_{\nu \in \Phi^{\alpha}  }\sum_{\xi_{\nu}, j_{\xi_{\nu}}} M_f^{\alpha}[\sqrt{\Pi_a}]^{\xi_{\omega} \ \xi_{\nu}}_{j_{\xi_{\omega}}j_{\xi_{\nu}}}M_f^{\alpha}[V'[(a,n)]]^{\xi_{\nu} \ \xi_{\rho}}_{j_{\xi_{\nu}}j_{\xi_{\rho}}}=\sum_{\nu \in \Phi^{\alpha}  }\sum_{ \xi_{\nu}, j_{\xi_{\nu}}}\frac{d_{\omega}d_{\nu}}{(n-1)d_{\alpha}}\sum_{k_{\alpha}}\psi_R^{\omega}[(a,n-1)]^{\xi_{\omega} \ \alpha}_{j_{\xi_{\omega}}k_{\alpha}}\psi^{\rho}_R[(a,n-1)]^{\alpha \ \xi_{\nu}}_{k_{\alpha}j_{\xi_{\nu}}}\times \\
		 & \times \frac{\sqrt{m_{\nu}m_{\rho}}}{m_{\alpha}}\sum_{l_{\alpha}}\psi_R^{\nu}[(a,n-1)]^{\xi_{\omega} \ \alpha}_{j_{\xi_{\nu}}l_{\alpha}}\psi^{\rho}_R[(a,n-1)]^{\alpha \ \xi_{\rho}}_{l_{\alpha}j_{\xi_{\rho}}}                                                                                                                                                                                                                                                                                                    \\
		 & =\frac{\sqrt{m_{\rho}d_{\omega}}}{(n-1)m_{\alpha}d_{\alpha}}\sum_{\nu, \xi_{\nu}, j_{\xi_{\nu}}}\sqrt{m_{\nu}d_{\nu}}\sum_{k_{\alpha},l_{\alpha}}\delta_{k_{\alpha}l_{\alpha}} \psi^{\omega}_R[(a,n-1)]_{j_{\xi_{\omega}}k_{\alpha}}^{\xi_{\omega} \ \alpha} \psi^{\rho}_R[(a,n-1)]^{\alpha \ \xi_{\rho}}_{l_{\alpha} j_{\xi_{\rho}}}                                                                                                                                                                              \\
		 & =\frac{\sqrt{m_{\rho}d_{\omega}}\sum_{\nu}\sqrt{m_{\nu}d_{\nu}}}{(n-1)d_{\alpha}m_{\alpha}}\sum_{k_{\alpha}}\psi_R^{\omega}[(a,n-1)]^{\xi_{\omega} \ \alpha}_{j_{\xi_{\omega}} k_{\alpha}}\psi^{\rho}_{R}[(a,n-1)]^{\alpha \ \xi_{\rho}}_{k_{\alpha}j_{\xi_{\rho}}}.
	\end{align}
	Having the above expression we are in position to evaluate trace $\tr\left(M_{f}^{\alpha }[\sqrt{\Pi _{a}}]M_{f}^{\alpha }[V^{\prime }[(a,n)]]\right)$. We have
	\begin{align}
		\tr\left(M_{f}^{\alpha }[\sqrt{\Pi _{a}}]M_{f}^{\alpha }[V^{\prime }[(a,n)]]\right)=\frac{\sum_{\nu \in \Phi^{\alpha}}\sqrt{m_{\nu}d_{\nu}}}{(n-1)m_{\alpha}d_{\alpha}}\sum_{\omega \in \Phi^{\alpha}}\sqrt{m_{\omega}d_{\omega}}\sum_{k_{\alpha}}\delta_{k_{\alpha}k_{\alpha}}=\frac{%
			1}{(n-1)m_{\alpha }}\left(\sum_{\nu \in \Phi^{\alpha}  }\sqrt{m_{\nu }d_{\nu }}%
		\right)^{2}.
	\end{align}
	Now, we compute the case when $h(\alpha)=d$ and an irrep $M_f^{\alpha}$ of the algebra $\A$ has a form presented in~\eqref{notfull}. In this case we consider only such irreps $\nu \in \Phi^{\alpha}$ for which $\nu \neq \theta$:
	\begin{align}
		\sum_{\substack{\nu \in \Phi^{\alpha}                                                                                                                                                \\ \nu\neq \theta}}\sum_{\xi_{\nu}, j_{\xi_{\nu}}} M_f^{\alpha}[\sqrt{\Pi_a}]^{\xi_{\omega} \ \xi_{\nu}}_{j_{\xi_{\omega}}j_{\xi_{\nu}}}M_f^{\alpha}[V'[(a,n)]]^{\xi_{\nu} \ \xi_{\rho}}_{j_{\xi_{\nu}}j_{\xi_{\rho}}}&=\frac{\sum_{\substack{\nu \in \Phi^{\alpha}\\ \nu\neq \theta}}\sqrt{m_{\nu}d_{\nu}}}{\sqrt{(n-1)d_{\alpha}-d_{\theta}}\sqrt{(n-1)d_{\alpha}}}\frac{\sqrt{d_{\omega}m_{\rho}}}{m_{\alpha}}\times\\
		 & \times\sum_{k_{\alpha}} \psi_R^{\omega}[(a,n-1)]^{\xi_{\omega} \ \alpha}_{j_{\xi_{\omega}} k_{\alpha}}\psi^{\rho}_{R}[(a,n-1)]^{\alpha \ \xi_{\rho}}_{k_{\alpha} j_{\xi_{\rho}}}.
	\end{align}
	Computing the trace from the above expression we have
	\begin{align}
		\tr\left(M_{f}^{\alpha }[\sqrt{\Pi _{a}}]M_{f}^{\alpha }[V^{\prime }[(a,n)]]\right) & =\frac{\sum_{\substack{\nu \in \Phi^{\alpha} \\ \nu\neq \theta}}\sqrt{m_{\nu}d_{\nu}}}{\sqrt{(n-1)d_{\alpha}-d_{\theta}}\sqrt{(n-1)d_{\alpha}}}\frac{\sum_{\substack{\omega \in \Phi^{\alpha}\\ \omega\neq \theta}}\sqrt{d_{\omega}m_{\omega}}}{m_{\alpha}}\sum_{k_{\alpha}}\delta_{k_{\alpha}k_{\alpha}}\\
		                                                                                    & =\frac{                                      %
		1}{\sqrt{(n-1)d_{\alpha }-d_{\theta }}}\frac{\sqrt{d_{\alpha }}}{\sqrt{(n-1)}%
		}\frac{1}{m_{\alpha }}\left(\sum_{\substack{\nu \in \Phi^{\alpha}                                                                  \\ \nu\neq \theta}}\sqrt{m_{\nu }d_{\nu }}\right)^{2}.
	\end{align}
	This finishes the proof.
\end{proof}

From this we deduce the value of the trace over full Hilbert space $\mathcal{H}=(\mathbb{C}^d)^{\otimes n}$ not only in a particular irrep $M_f^{\alpha}$ of the algebra $\A$. Namely we have the following:

\begin{theorem}
	\label{thm:trH}
	For numbers $n\in \mathbb{N}$, $d\geq 2$ in the algebra $\A$ we have
	\begin{equation}
		\label{trH}
		\tr_{\mathcal{H}}\left(\sqrt{\Pi _{a}}V^{\prime }[(a,n)]\right)
		=\sum_{\alpha :h(\alpha )<d}\frac{1}{n-1}\left(\sum_{\nu \in \alpha}
		\sqrt{m_{\nu }d_{\nu }}\right)^{2}+\sum_{\alpha :h(\alpha )=d}\frac{1}{\sqrt{%
		(n-1)d_{\alpha }-d_{\theta }}}\frac{\sqrt{d_{\alpha }}}{\sqrt{(n-1)}}%
		\left(\sum_{\substack{\nu \in \Phi^{\alpha}\\ \nu\neq \theta}}\sqrt{m_{\nu }d_{\nu }}\right)^{2}.
	\end{equation}
\end{theorem}

\begin{proof}
	To prove the statement of this theorem we have consider two cases when $h(\alpha)<d$, then we use expression~\eqref{full}, and when $h(\alpha)=d$, then we use expression~\eqref{notfull}. Since both equations are evaluated in a given irrep $M_f^{\alpha}$ of the algebra $\A$ we need to sum up all such contributions, everyone with multiplicity $m_{\alpha}$. This leads us to expression~\eqref{trH} and finishes the proof.
\end{proof}

\begin{lemma}
	\label{l:lemma9}
	In the qubit case ($d=2$) the expression \eqref{trH} takes the form
	\be
	\tr_{\mathcal{H}}\left(\sqrt{\Pi _{a}}V^{\prime }[(a,n)]\right)=\frac{1}{N}\sum_{l=0}^k \sqrt{\frac{(N+1-l)(l+1)}{N+1}}\left((N-2l+1)\sqrt{\frac{1}{N+1}{N+1 \choose l}} + (N-2l-1)\sqrt{\frac{1}{N+1}{N+1 \choose l+1}}\right)^2,
	\ee
	where $N=n-1$ and $a=1,\dots,N.$
\end{lemma}
\begin{proof}
	In qubit case, only two types of Young diagrams $\lambda_\alpha$ for $\alpha\in\widehat{S(n-1)}$ are possible: either $\lambda_\alpha=(n-1-l, l)$ or $\lambda_\alpha=(n-1).$ We can denote the respective irreps accordingly to the number of the rows i.e. $\lambda_\alpha=(n-1,l) \coloneqq \alpha_l$.

	In  the expression \eqref{trH} the only irreps $\nu\in\alpha,\;\alpha_l\in\widehat{S(n-2)},h(\alpha)=1$ such that $\nu\in\widehat{S(n-1)}$ are $\nu_0$ and $\nu_1$. Similarly, the irreps $\nu\in\alpha_l,\;\alpha_l\in\widehat{S(n-2)},h(\alpha_l)=2$ are $\nu_{l}$ and $\nu_{l+1}$, unless for $\alpha_l=(n-2-l, l)$ we have  $n-2=2l$ and in such case only $\nu_l$ is present. The expression \eqref{trH} becomes
	\be
	\label{eq:odd}
	\tr_{\mathcal{H}}\left(\sqrt{\Pi _{a}}V^{\prime }[(a,n)]\right)=\frac{1}{n-1}
	\left(\sqrt{m_{\nu_0}d_{\nu_0 }} + \sqrt{m_{\nu_1}d_{\nu_1 }}\right)^{2}
	+\sum_{l=1}^k\frac{1}{\sqrt{%
	(n-1)d_{\alpha_l }-d_{\theta }}}\frac{\sqrt{d_{\alpha_l }}}{\sqrt{(n-1)}}%
	\left(\sqrt{m_{\nu_l }d_{\nu_l }} + \sqrt{m_{\nu_{l+1} }d_{\nu_{l+1} }}\right)^{2}.
	\ee
	for odd  $n$, where $k=\left\lfloor \frac{n-2}{2}\right\rfloor$ and
	\begin{align}
		\label{eq:even}
		\tr_{\mathcal{H}}\left(\sqrt{\Pi _{a}}V^{\prime }[(a,n)]\right)=\frac{1}{n-1}
		\left(\sqrt{m_{\nu_0}d_{\nu_0 }} + \sqrt{m_{\nu_1}d_{\nu_1 }}\right)^{2}
		+ & \sum_{l=1}^{k-1}\frac{1}{\sqrt{                                                %
		(n-1)d_{\alpha_l }-d_{\theta }}}\frac{\sqrt{d_{\alpha_l }}}{\sqrt{(n-1)}}%
		\left(\sqrt{m_{\nu_l }d_{\nu_l }} + \sqrt{m_{\nu_{l+1} }d_{\nu_{l+1} }}\right)^{2} \\
		  & \frac{1}{\sqrt{                                                                %
		(n-1)d_{\alpha_k }-d_{\theta }}}\frac{\sqrt{d_{\alpha_k }}}{\sqrt{(n-1)}}%
		\left(\sqrt{m_{\nu_k }d_{\nu_k }}\right)^{2}
	\end{align}
	for even $n$ and $k=\frac{n-2}{2}$.

	The expression for $m_\alpha$ and $d_\alpha, \alpha=(n-l-1, l)$, for Young diagrams with at most two rows are given by \cite{fulton_harris}
	\be
	d_\alpha =\binom{n-1}{l}-\binom{n-1}{l-1}=\frac{(n-2l)}{n}\binom{n}{l},\quad m_\alpha=(n-2l)
	\label{eq:m_d_alpha}
	\ee
	Moreover, in case of $\theta$ that has three rows $\theta=(n-2-l,l,1)$ the value for $d_\theta$ can be obtained by hook-length formula:
	\be
	d_\theta=\frac{(n-1)!}{\Pi_{i,j} h_\theta(i,j)}
	\ee
	where $h_\theta(i,j)$ is the sum of the number of boxes in $i$th row from $j$th box to the end of the row and the number of boxes in $j$th column after $i$th box, which is so-called hook length. Considering $\alpha_l = (n-2-l, l)$ and $\theta = (n-2-l, l, 1)$  the ony hooks that differ are hooks in the points $(1,1)$ and $(2,1)$. Denoting the product of the common hooks by $R$ we have

	\be
	d_\theta=\frac{(n-1)!}{(n-l)(l+1)R} = (n-1)\frac{(n-1-l)l}{(n-l)(l+1)} \frac{(n-2)!}{(n-1-l)lR} = (n-1)\frac{(n-1-l)l}{(n-l)(l+1)} d_{\alpha_l},
	\label{eq:d_theta}
	\ee

	Therefore we have

	\begin{align}
		\frac{1}{\sqrt{
		(n-1)d_{\alpha_l }-d_{\theta }}}\frac{\sqrt{d_{\alpha_l }}}{\sqrt{(n-1)}}%
		 & = \frac{1}{\sqrt{
				(n-1)-(n-1)\frac{(n-1-l)l}{(n-l)(l+1)} }}\frac{1}{\sqrt{(n-1)}} \\
		 & = \frac{1}{n-1}\sqrt{\frac{(n-l)(l+1)}{n}}
	\end{align}
	and the expressions \eqref{eq:odd} and \eqref{eq:even} become
	\begin{align}
		\label{eq:odd2}
		\tr_{\mathcal{H}}\left(\sqrt{\Pi _{a}}V^{\prime }[(a,n)]\right) & =\frac{1}{n-1}
		\left(n\sqrt{\frac{1}{n}\binom{n}{0}} + \left(n-2\right)\sqrt{\frac{1}{n}\binom{n}{1}} \right)^{2}                        \\
		                                                                & +\frac{1}{n-1}\sqrt{\frac{(n-l)(l+1)}{n}}\sum_{l=1}^{k} %
		\left((n-2l)\sqrt{\frac{1}{n}\binom{n}{l}} + (n-2(l+1))\sqrt{\frac{1}{n}\binom{n}{l+1}}   \right)^{2}.
	\end{align}
	for odd  $n$, where $k=\left\lfloor \frac{n-2}{2}\right\rfloor$ and
	\begin{align}
		\label{eq:even2}
		\tr_{\mathcal{H}}\left(\sqrt{\Pi _{a}}V^{\prime }[(a,n)]\right) & =\frac{1}{n-1}
		\left(n\sqrt{\frac{1}{n}\binom{n}{0}} + \left(n-2\right)\sqrt{\frac{1}{n}\binom{n}{1}} \right)^{2}                          \\
		                                                                & +\frac{1}{n-1}\sqrt{\frac{(n-l)(l+1)}{n}}\sum_{l=1}^{k-1} %
		\left((n-2l)\sqrt{\frac{1}{n}\binom{n}{l}} + \left(n-2(l+1)\right)\sqrt{\frac{1}{n}\binom{n}{l+1}}   \right)^{2}            \\
		                                                                & +                                                         %
		\frac{1}{n-1}\sqrt{\frac{(n-k)(k+1)}{n}}\left( \frac{1}{n}\binom{n}{k} \right)^{2}
	\end{align}
	for even $n$ and $k=\frac{n-2}{2}$.

	Setting $k=\left\lfloor\frac{n-2}{2}\right\rfloor$ and $N=n-1$ we can see that both expressions simplify to one expression, no matter the parity of $n-2$

	\be
	\tr_{\mathcal{H}}\left(\sqrt{\Pi _{a}}V^{\prime }[(a,n)]\right)=\frac{1}{N}\sum_{l=0}^k \sqrt{\frac{(N+1-l)(l+1)}{N+1}}\left((N-2l+1)\sqrt{\frac{1}{N+1}{N+1 \choose l}} + (N-2l-1)\sqrt{\frac{1}{N+1}{N+1 \choose l+1}}\right)^2,
	\ee
	which completes the proof.
\end{proof}

\section{Detailed analysis of the recycling protocol for port-based teleportation}
\label{rec_PBT}
Having presented the analysis of the square-root measurements in the previous section we are in position to apply our findings to the analysis of the recycling protocol. Here we focus on  how much the remaining ports degradate after a sigle round of the recycling scheme. Similarly to~\cite{strelchuk_generalized_2013}, we investigate the non-optimal deterministic PBT, when $O_A=\mathbf{1}_A$ in equation~\eqref{resource}, and present detailed discussion for the optimal PBT in the respective appendices.

However, before we proceed further  we need to introduce and fix some notation. By $|\psi_{in}\>=|\psi^+_{A_0B_0}\>\otimes |\Phi\>_{AB}$ we denote the total state of the resource state and state to be teleported before parties run the protocol. Next, by $|\psi^{(a)}_{id}\>$ we denote the total state after the ideal process of teleportation to $a-$th port:
\begin{equation}
	|\psi^{(a)}_{id}\>=|\psi^+_{A_0A_a}\>|\psi^+_{B_0B_a}\>\otimes \left(\bigotimes_{\substack{j=1\\j\neq a}}^N |\psi_{A_jB_j}^+\>\right).
\end{equation}
Finally, by $|\psi^{(a)}_{out}\>$ we denote the total state after application of a measurement $\widetilde{\Pi}_a^{AA_0}$:
\begin{equation}
	\label{eq8a}
	\begin{split}
		|\psi_{out}^{(a)}\>=\frac{\left(\sqrt{\widetilde{\Pi}_a^{AA_0}}\otimes \mathbf{1}_{B_0B}\right)|\psi_{in}\>_{A_0B_0AB}}{\left|\left|\left(\sqrt{\widetilde{\Pi}_a^{AA_0}}\otimes \mathbf{1}_{B_0B}\right)|\psi_{in}\>_{A_0B_0AB}\right|\right|_2}.
	\end{split}
\end{equation}
Now, we see that to describe qualitatively the efficiency of the recycling scheme we have to compute the average fidelity $F(\mathcal{P}_{rec}(N,d,1))$ between the state of all the ports $|\psi^{(a)}_{out}\>$ after application of a measurement $\widetilde{\Pi}_a^{AA_0}$ and the idealised situation, where the teleportation is carried out without any disturbance and state of the ports is $|\psi^{(a)}_{id}\>.$ Now, with the number of ports growing the fidelity of the teleported state goes to 1, since we perform PBT~\cite{majenz2}. If the same situation we observe for the fidelity $F(\mathcal{P}_{rec}(N,d,1))$  it means that the real state is close to the idealised one. From this one can deduce that remaining ports, those except $a$-th one, do not suffer too much from the measurement $\Pi_a^{AA_0}$. Therefore, our next goal is to find expression for the mentioned fidelity $F(\mathcal{P}_{rec}(N,d,1))$.  We start from defining corresponding density matrices $\psi^{(a)}_{out}:=|\psi^{(a)}_{out}\>\<\psi^{(a)}_{out}|$ and $\psi^{(a)}_{id}:=|\psi^{(a)}_{id}\>\<\psi^{(a)}_{id}|$ for which the fidelity $F(\mathcal{P}_{rec}(N,d,1))$ is
\begin{equation}
	\label{PF}
	F(\mathcal{P}_{rec}(N,d,1))=\sum_{a=1}^Np_aF\left(\psi^{(a)}_{out},\psi^{(a)}_{id}\right)=\frac{1}{d^{N+1}}\sum_{a=1}^N\tr( \widetilde{\Pi}_a^{A_0A})F\left(\psi^{(a)}_{out},\psi^{(a)}_{id}\right),
\end{equation}
since $p_a=\frac{1}{d^{N+1}}\tr (\widetilde{\Pi}_a^{A_0A})$. Having the above we are ready to show connection of $F(\mathcal{P}_{rec}(N,d,1))$ with signal states and Alice's measurements  in arbitrary $d$ (see also Section 2 from Supplementary Materials of~\cite{strelchuk_generalized_2013}):

\begin{lemma}
	\label{thm1_rec}
	The fidelity $F(\mathcal{P}_{rec}(N,d,1))$ in the recycling scheme, with $N$ ports, each of dimension $d$, after one round of teleportation is the following:
	\begin{equation}
		\label{thmeq}
		F(\mathcal{P}_{rec}(N,d,1))=\frac{N}{d}\frac{\sqrt{\tr(\widetilde{\Pi}_N^{A_0A})}}{\sqrt{d^{N+1}}}\left|\tr\left(\sigma_{A_0A_N}\sqrt{\widetilde{\Pi}_N^{A_0A}}\right)\right|,
	\end{equation}
	where $\sigma_N,\widetilde{\Pi}_N^{A_0A}$ are respectively the signal state and the measurement corresponding to index $a=N$ in equations~\eqref{eq:signals} and~\eqref{eq:measurements}.
\end{lemma}
The proof of the above Lemma is located in Appendix~\ref{AppA0}.

Now we show that indeed the expression from the lemma above is well defined. Namely, we prove that $F(\mathcal{P}_{rec}(N,d,1))\leq 1$.
\begin{remark}
	\label{Frec_bound}
	The fidelity $F(\mathcal{P}_{rec}(N,d,1))$ in the one round of the recycling protocol $\mathcal{P}_{rec}(N,d,1)$, with $N$ ports, each of dimension $d$, satisfies the following bound
	\begin{equation}
		\label{eq:Frec_bound}
		F(\mathcal{P}_{rec}(N,d,1))\leq 1.
	\end{equation}
\end{remark}

Indeed, applying the Schwarz inequality for  the scalar product of operators $\sigma_{A_0A_N}$ and $\sqrt{\widetilde{\Pi}_N^{A_0A}}$ in equation~\eqref{thmeq} in Theorem~\ref{thm1_rec}, we bound $F(\mathcal{P}_{rec})$ as
\begin{equation}
	\label{F1}
	\begin{split}
		F(\mathcal{P}_{rec}(N,d,1))=\frac{N}{d}\frac{\sqrt{\tr(\widetilde{\Pi}_N^{A_0A})}}{\sqrt{d^{N+1}}}\left|\tr\left(\sigma_{A_0A_N}\sqrt{\widetilde{\Pi}_N^{A_0A}}\right)\right|\leq \frac{N}{d}\frac{\sqrt{\tr(\widetilde{\Pi}_N^{A_0A})}}{\sqrt{d^{N+1}}}\sqrt{\tr(\widetilde{\Pi}_N^{A_0A})\tr(\sigma_{A_0A_N}^2)}=\frac{N}{d^{N+1}}\tr(\widetilde{\Pi}_N^{A_0A}),
	\end{split}
\end{equation}
since due to~\eqref{eq:signals} we have $\tr(\sigma_{A_0A_N}^2)=(1/d^{N-1})\tr(\sigma_{A_0A_N})=1/d^{N-1}$. The above requires an additional justification. Due to definitions from~\eqref{eq:measurements}, we have that $\operatorname{supp}(\sum_a \Pi_a^{A_0A})=\operatorname{supp}(\rho)=\mathcal{H}_{\mathcal{M}}$. Next, due to~\eqref{Delta} we know that $\operatorname{supp}(\Delta)\equiv \mathcal{H}_{\mathcal{S}}=\mathbf{1}_{(\mathbb{C}^d)^{\otimes N+1}}\ominus \mathcal{H}_{\mathcal{M}}$. These relations imply that $\operatorname{supp}(\sigma_{A_0A_N})\subset \mathcal{H}_{\mathcal{M}}\perp \mathcal{H}_{\mathcal{S}}$, so $\tr(\sigma_{A_0A_N} m)=0$, for all elements $m\in \mathcal{H}_{\mathcal{S}}$.

Now we have to evaluate $\tr(\widetilde{\Pi}_N^{A_0A})$. First, let us recall that $(\mathbb{C}^d)^{\otimes N+1}\equiv \mathcal{H}=\mathcal{H}_{\mathcal{M}}\oplus \mathcal{H}_{\mathcal{S}}$, so we have the following relations
\begin{equation}
	\label{trPiN}
	\begin{split}
		&\dim \mathcal{H}=\dim\mathcal{H}_{\mathcal{M}}+ \dim \mathcal{H}_{\mathcal{S}}\\
		&\sum_{a=1}^N \tr(\widetilde{\Pi}_a^{A_0A})= \sum_{a=1}^N \tr(\Pi_a^{A_0A})+\tr(\Delta)\\
		& N\tr(\widetilde{\Pi}_N^{A_0A})=N\tr(\Pi_N^{A_0A})+d^{N+1}-N\tr(\Pi_N^{A_0A})\\
		&\tr(\widetilde{\Pi}_N^{A_0A})=\frac{d^{N+1}}{N},
	\end{split}
\end{equation}
where in the third line we use independence of trace with respect to index $a$.  Finally, substituting~\eqref{trPiN} to~\eqref{F1} we get the statement presented in~\eqref{eq:Frec_bound}. This finishes the proof.

Once can see that expression~\eqref{thmeq} from Lemma~\ref{thm1_rec} is given in terms of the operators acting on a very large  space $(\mathbb{C}^d)^{\ot (n+1)}$. Moreover, one must know the support of the operator $\rho$ from~\eqref{eq:measurements} to compute all the quantities interesting in PBT. These two facts make the analysis of the performance of the recycling protocol almost intractable without further simplifications.
Because of that our goal will be to find an explicit expression for $F(\mathcal{P}_{rec}(N,d,1))$ by evaluating trace in~\eqref{thmeq} by exploiting existing symmetries discussed in Section~\ref{dPBT}. To obtain such result first we need to learn about the interior structure of POVM operators $\Pi_a^{A_0A}$, which will allow us to compute their square root and the overlap with the signal states. Below we present explicit equation for $F(\mathcal{P}_{rec}(N,d,1))$ depending on group-theoretic quantities describing permutation groups $S(N-1)$ and $S(N)$. These quantities (multiplicities, dimensions of irreps) are effectively computable even for a large number of ports $N$ and local dimension $d$. This however can be done for example by such packages like SAGE~\cite{sage}. We start from the following theorem.

\begin{theorem}
	\label{th:qubit_f}
	\label{expplicit}
	The fidelity $F(\mathcal{P}_{rec}(N,d,1))$ in the one round of the recycling protocol $\mathcal{P}_{rec}(N,d,1)$, with $N$ ports, each of dimension $d$, reads
	\begin{equation}
		\label{eq:qubit_f0}
		F(\mathcal{P}_{rec}(N,d,1))=\frac{\sqrt{N}}{d^{N+1}}\left[\sum_{\alpha :h(\alpha )<d}\frac{1}{N}\left(\sum_{\nu \in \alpha}
		\sqrt{m_{\nu }d_{\nu }}\right)^{2}+\sum_{\alpha :h(\alpha )=d}\frac{1}{\sqrt{%
		Nd_{\alpha }-d_{\theta }}}\frac{\sqrt{d_{\alpha }}}{\sqrt{N}}%
		\left(\sum_{\nu \neq \theta }\sqrt{m_{\nu }d_{\nu }}\right)^{2}\right].
	\end{equation}
	By $m_{\nu}$ we denote the multiplicity of irreps of $S(N)$ in the Schur-Weyl duality, by $d_{\alpha}, d_{\alpha}$ dimensions of irreps $S(N-1)$ and $S(N)$ respectively in the Schur-Weyl duality. The index $d_{\theta}$ denotes irrep dimension of $S(N)$ of height $d+1$ obtained from irrep of $S(N-1)$ whose height is $d$. If there are no such irreps, then we set $d_\theta\equiv 0$.
\end{theorem}

\begin{proof}
	To prove the expression~\eqref{eq:qubit_f0} in Theorem~\ref{expplicit} it is enough to apply the statement of Theorem~\ref{thm:trH}. Knowing that $\widetilde{\Pi_a}=\Pi_a+\frac{1}{N}\Delta$ together with $\tr (\widetilde{\Pi_a})=d^{N+1}/N$ from expression~\eqref{trPiN}, and fact that $\operatorname{supp}(\Delta)\perp \operatorname{supp}(\sigma_a)$ for~$1\leq a\leq n-1$, we have
	\begin{align}
		F(\mathcal{P}_{rec}) & =\frac{N}{d}\frac{\sqrt{\tr(\Mn)}}{\sqrt{d^{N+1}}}\left|\tr\left(\sigma_N\sqrt{\Mn}\right)\right|=\frac{N}{d^{N+1}}\frac{\sqrt{\tr(\Mn)}}{\sqrt{d^{N+1}}}\left|\tr\left(V'\sqrt{\Mn}\right)\right|=\frac{\sqrt{N}}{d^{N+2}}\left|\tr\left(V'\sqrt{\Pi_N}\right)\right| \\
		                     & =\frac{\sqrt{N}}{d^{N+1}}\left[\sum_{\alpha :h(\alpha )<d}\frac{1}{n-1}\left(\sum_{\nu \in \alpha}
		\sqrt{m_{\nu }d_{\nu }}\right)^{2}+\sum_{\alpha :h(\alpha )=d}\frac{1}{\sqrt{%
		(n-1)d_{\alpha }-d_{\theta }}}\frac{\sqrt{d_{\alpha }}}{\sqrt{(n-1)}}%
		\left(\sum_{\substack{\nu \in \Phi^{\alpha}                                                                                                                                                                                                                                                   \\ \nu\neq \theta}}\sqrt{m_{\nu }d_{\nu }}\right)^{2}\right].
	\end{align}
\end{proof}
In the case of qubits, when $d=2$, we can rewrite the statement of Theorem~\ref{expplicit} in much more appealing form, depending only on number of ports $N$ exploited in PBT scheme. This is possible by direct application of Lemma~\ref{l:lemma9} from Section~\ref{StructurePOVMs}.

\begin{lemma}
	\label{Fqubits}
	The fidelity $F(\mathcal{P}_{rec}(N,2,1))$ in the one round of the qubit recycling protocol $\mathcal{P}_{rec}(N,2,1)$, with $N$ ports, reads as
	\be
	F(\mathcal{P}_{rec}(N,2,1))=\frac{\sqrt{N}}{2^{N+1}}\sum_{l=0}^k \sqrt{\frac{(N+1-l)(l+1)}{N+1}}\left((N-2l+1)\sqrt{\frac{1}{N+1}{N+1 \choose l}} + (N-2l-1)\sqrt{\frac{1}{N+1}{N+1 \choose l+1}}\right)^2,
	\ee
	where $k =\left \lceil \frac{N}{2}-1\right\rceil$.
\end{lemma}

\begin{figure}[t]
	\includegraphics[width=0.45\textwidth]{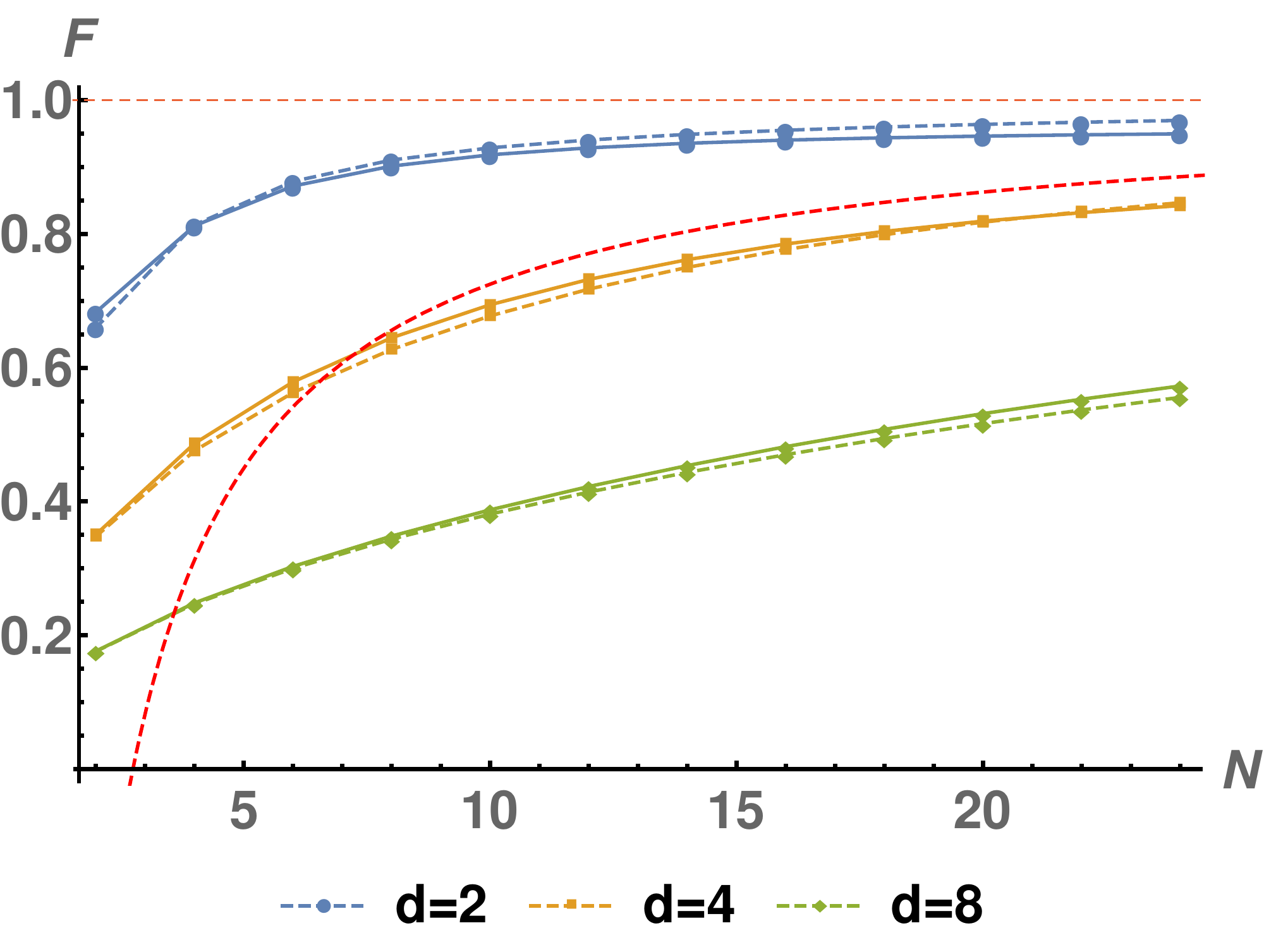}
	\includegraphics[width=0.45\textwidth]{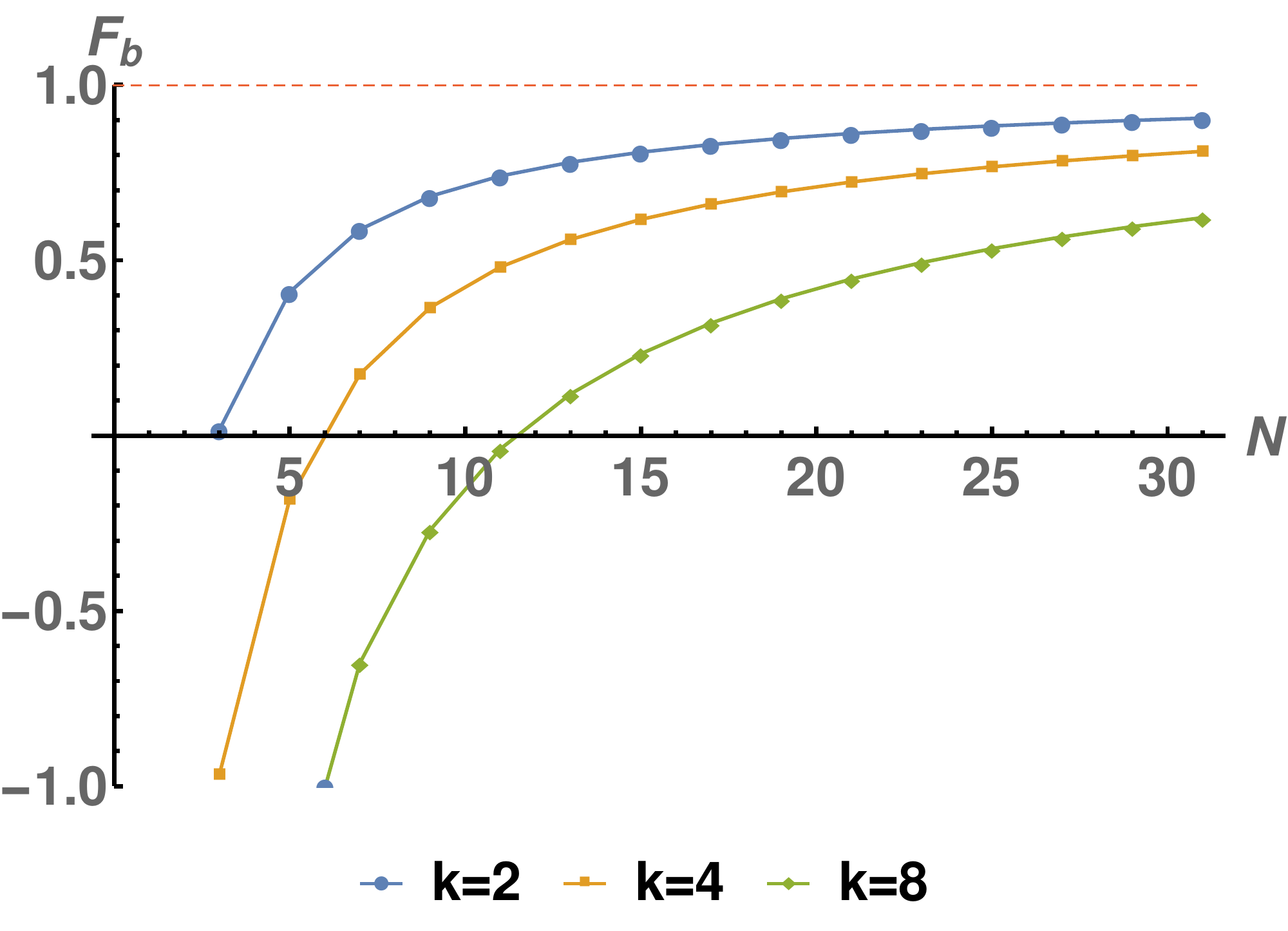}
	\caption{The left panel presents values of fidelity $F(P_{rec}(N,d,1))$ evaluated for non- (dashed lines) and optimal (solid lines) PBT. From these plots we see that the resource state for optimal PBT is not necessarily better for recycling protocol than the resource state for its non-optimal counterpart. In fact, for $d=2$ values of $F(P_{rec}(N,2,1))$ for optimal version are even worse than for non-optimal one. The numerical values for $d>2$ have been obtained from Theorems~\ref{th:qubit_f} and~\ref{F_rec_optimal}. In the qubit case we have used Lemmas~\ref{Fqubits} and~\ref{lem:f_opt}. The red dashed line shows the lower bound on fidelity in the qubit case given by Eq.~\eqref{eq1} up to $\mathcal{O}(1/N^2)$ part. The right panel presents the lower bound for $F(P_{rec}(N,2,k))$ (the qubit case) and different number of teleportation rounds. From the plot we see that $F(P_{rec}(N,2,k))$ is relatively high even for not too large number of ports.
	}
	\label{fig:fid}
\end{figure}

Now, one could ask how the recycling protocol behaves when we consider optimised version of deterministic PBT. In this case measurements and the resource state is optimised by Alice simultaneously, and optimisation resulting in the following explicit form of the operation $O_A$ in equation~\eqref{resource} derived in~\cite{ishizaka_quantum_2009,StuNJP}:
\begin{equation}
	\label{expOA}
	O_A=\sqrt{d^N}\sum_{\mu \vdash N}\frac{v_{\mu}}{\sqrt{d_{\mu}m_{\mu}}}P_{\mu},
\end{equation}
where $v_{\mu}\geq 0$ are entries of a normalised eignevector corresponding to a maximal eigenvalue of the teleportation matrix $M_F$ used for computation of entanglement fidelity in OPBT~\cite{StuNJP}, and $P_{\mu}$ is a Young projector defined in~\eqref{Yng_proj}. Having that we are in position to generalise Lemma~\ref{thm1_rec}, Theorem~\ref{expplicit}, and Lemma~\ref{Fqubits} to the optimal case. Lemma analogous to Lemma~\ref{thm1_rec} is allocated in Appendix~\ref{rec_OdPBT}.
\begin{theorem}
	\label{F_rec_optimal}
	The fidelity $F(\mathcal{P}_{rec}(N,d,1))$ in the recycling scheme for the optimal deterministic PBT scheme, with $N$ ports, each of dimension $d$, after one round of teleportation is the following:
	\begin{equation}
		\begin{split}
			F(\mathcal{P}_{rec}(N,d,1))&=\frac{1}{d^{1/2}}\sum_{\alpha \vdash N-1}\sum_{\mu \in \alpha}\frac{v_{\alpha}v_{\mu}}{m_{\alpha}^{1/2}}\frac{\sum_{\substack{\nu \neq \theta\\
					\nu\in \alpha}}\sqrt{m_{\nu}d_{\nu}}}{\sqrt{Nd_{\alpha}-d_{\theta}}},
			\label{eq:f_opt}
		\end{split}
	\end{equation}
	where $v_{\alpha},v_{\mu}$ are the coefficients of operations $O_{\widetilde{A}},O_A$ given in~\eqref{expOA} for $N-1$ and $N$ ports respectively, for which Young frames are in the relation $\mu\in\alpha$. The numbers $m_{\alpha},m_{\nu}$ and $d_{\alpha},d_{\nu}$ denote multiplicities and dimensions of irreps of $S(N-1)$ and $S(N)$ respectively in the Schur-Weyl duality. Finally by $\theta$ we denote irreps of dimension $d_{\theta}$ of $S(N)$ belonging to the set $\Theta$ given through~\eqref{Theta}.
\end{theorem}
The proof of the above theorem is contained in Appendix~\ref{rec_OdPBT}. Similarly as it was for Theorem~\ref{expplicit} we present the general statement of Theorem~\ref{F_rec_optimal} in the qubit case, where the final expression depends only on the total number of ports $N$. In this case all Young frames are up to two rows and they are of the form $(N-l,l)$, so the entries  entries $v_{\mu}$ are labelled by two indices $N,l$ as $v_l^{(N)}$.
\begin{lemma}
	\label{lem:f_opt}
	In the special case $d=2$ the expression~\eqref{eq:f_opt} reads
	\be
	\begin{split}
		&F(\mathcal{P}_{rec}(N,2,1)) =\\
		&=\frac{1}{\sqrt 2}\sum _{l=0}^{k }
		\frac{v_l^{(N-1)}\left({v_l^{(N)} + v_{l+1}^{(N)}}\right)}{N-2l}
		\left ((N-2l+1)\sqrt{\frac{{N+1 \choose l}}{N+1}} + (N-2l-1)\sqrt{\frac{{N+1 \choose l+1}}{N+1}} \right)^2
		\sqrt{\frac{(N-l+1)(l+1)}{(N-2l)(N+1){N\choose l}}},
	\end{split}
	\ee
	where $k=\left \lceil{\frac{N}{2}-1}\right \rceil$, $v_l^{(N)}$ is the coefficient of the operator $O_A$ associated with the irrep $\mu=(N-l,l)$ in the qubit case. If $l+1>N/2$, it is equal to 0, otherwise it is given by
	\be
	v_l^{(N)}=\begin{cases} (-1)^{\frac{N}{2}-l}
			{\left(\operatorname{sin}
		\frac{\left(\frac{N+2}{2}-l\right) N\pi}{N+2} - \operatorname{sin}\frac{\left(\frac{N}{2}-l\right) N\pi}{N+2}\right)}\bigg/{\operatorname{sin}\frac{ N\pi}{N+2}}\; & \textit{for even N}, \\
		(-1)^{\frac{N-1}{2}-l}\operatorname{sin} \frac{\left(\frac{N+1}{2}-l\right) N\pi}{N+2}\big/{\operatorname{sin}\frac{ N\pi}{N+2}}\;                                 & \textit{for odd N.}
	\end{cases}
	\label{eq:v_l}
	\ee
\end{lemma}
The proof of the above lemma is located in Appendix~\ref{AppE} and ~\ref{app:opbtd2}. The values of $F(P_{rec}(N,d,1))$ for various values of port dimension  and both variants of PBT are depicted in the left panel of Figure~\ref{fig:fid}.  We see from it that the bound from Remark~\ref{Frec_bound} is attained reasonably fast.

Having expressions for $F(P_{rec}(N,d,1))$ in non- and optimal case we can say how the recycling protocol behaves after many rounds $k$. To do so use exactly the same reasoning as in the proof of Lemma 2 in~\cite{strelchuk_generalized_2013}, since the proof is dimension independent. This leads us to the following lower bound:
\begin{equation}
	\label{eq:accumulation}
	F(P_{rec}(N,d,k))\geq 1-2k(1-F(P_{rec}(N,d,1))),
\end{equation}
which obviously tends to 1 when $F(P_{rec}(N,d,1))\rightarrow 1$. This also shows that the error is additive with the number of rounds $k$. The bound from~\eqref{eq:accumulation}, for various values of rounds, is depicted on the right panel of Figure~\ref{fig:fid}.

The above results show clearly that there is no connection between the fidelity $F(P_{rec}(N,d,1))$ an the type of PBT protocol.
First of all, notice that the resource states in non- and optimal PBT are in fact very different, when we compute fidelity between them.
Let us consider resource state $|\Phi^+\>_{AB}$ in non-optimal PBT, when $O_A=\mathbf{1}_A$ in~\eqref{resource}, and optimal PBT $|\Phi\>_{AB}$, where $O_A$ is given through~\eqref{expOA}. Having that we can formulate the following lemma:
\begin{lemma}
	\label{l:FPBT}
	The fidelity between the resource state in non-optimal and optimal PBT with $N$ ports, each of dimension $d$ is given as:
	\begin{equation}
		\label{FPBT}
		F(|\Phi^+\>_{AB},|\Phi\>_{AB})=\frac{1}{\sqrt{d^N}}\sum_{\mu \vdash N} v_{\mu}\sqrt{d_{\mu}m_{\mu}},
	\end{equation}
	where $v_{\mu}$ are entries of an eignevector corresponding to a maximal eigenvalue of the teleportation matrix, $m_{\mu},d_{\mu}$ denote multiplicity and dimension of irreps of $S(N)$ in the Schur-Weyl duality, and $P_{\mu}$ is a respective Young projector.
	In particular case of qubits, when $d=2$, the fidelity~\eqref{FPBT} is of the form:
	\begin{equation}
		\label{F_qubit}
		F(|\Phi^+\>_{AB},|\Phi\>_{AB})=\frac{1}{d^{N+1}}\sum_{l=0}^{\lfloor \frac{N}{2}\rfloor} v_l^{(N)} \frac{(N-2l+1)}{\sqrt{N+1}}\sqrt{{N+1 \choose l}}.
	\end{equation}
\end{lemma}
The proof of the above lemma is located in Appendix~\ref{AppB} with derivation of equivalent expression to~\eqref{F_qubit} in the picture of quantum angular momentum. 
\begin{figure}
	\centering
	\includegraphics[width=0.45\textwidth]{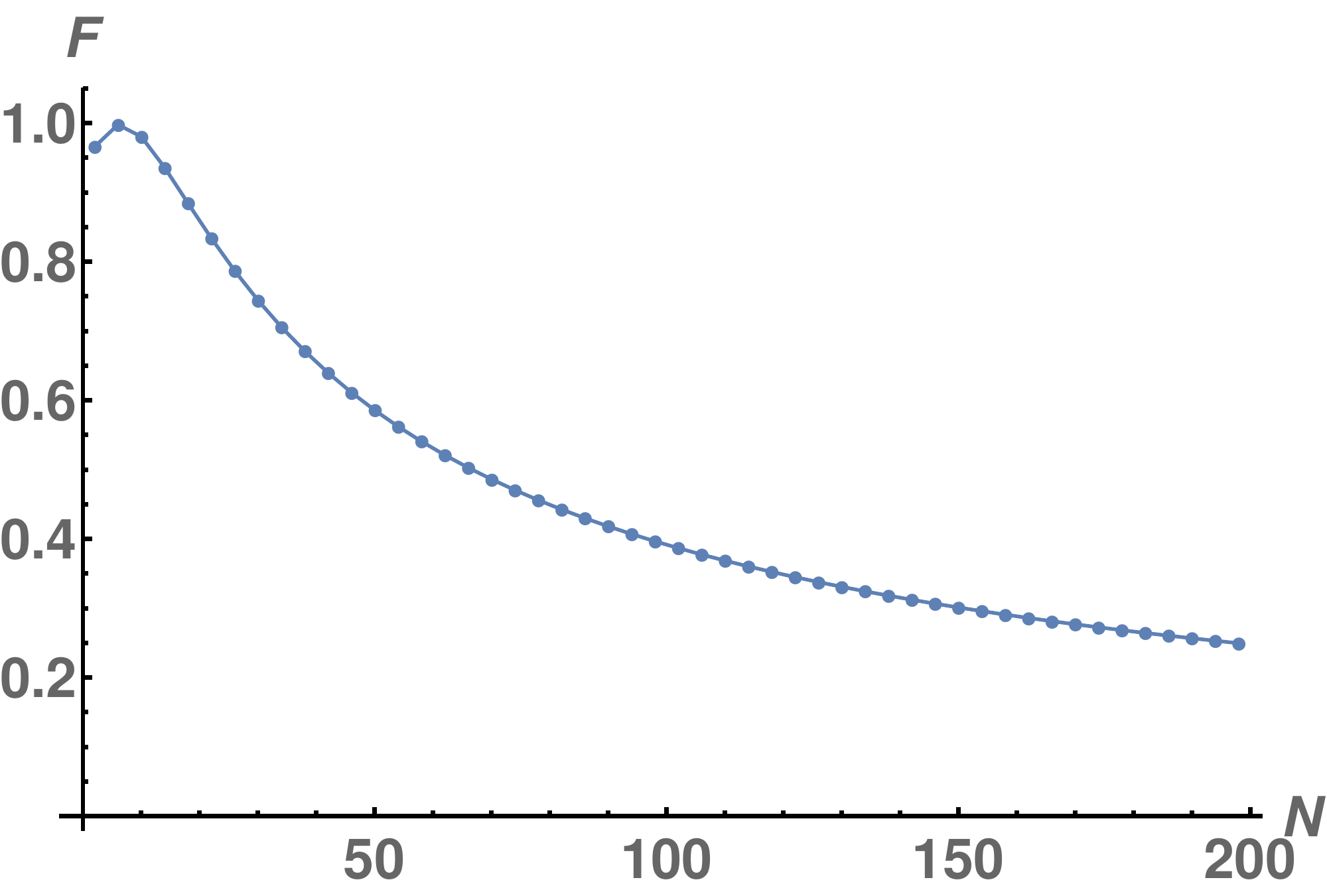}
	\includegraphics[width=0.45\textwidth]{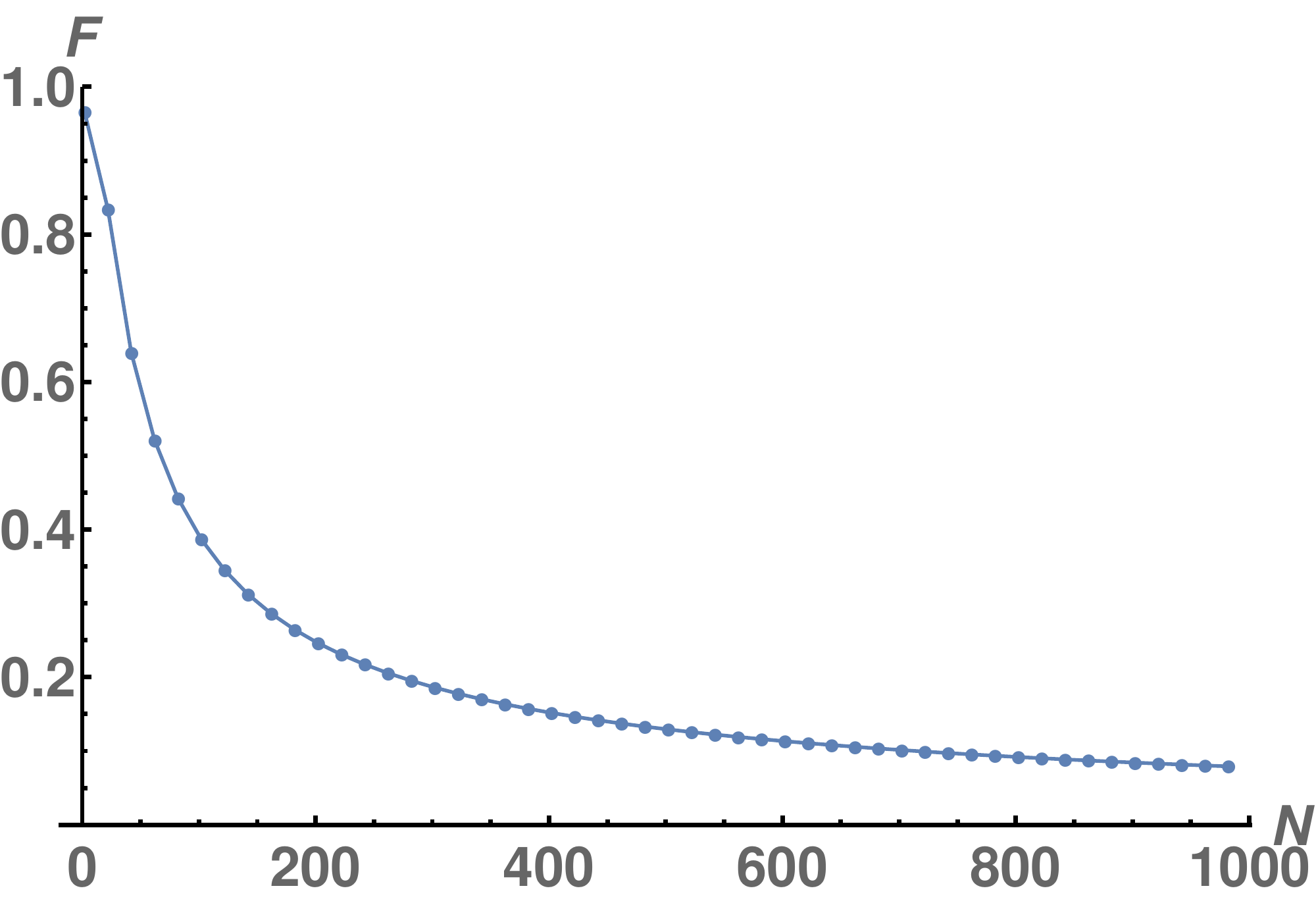}
	\caption{Figures depict the qubit fidelity between resource states in the non- and the optimal PBT for different maximal number of ports $N$. In the left panel we see that the mentioned fidelity is not a monotonic function for the small number of ports. Its maximal value which is $F=0.9977$ is attained for $N=6$. In the right panel we see monotonic behaviour for a large number of ports. In the asymptotic limit and the qubit case, the both states are orthogonal in the limit of $N\rightarrow \infty$.}
	\label{fig:my_label}
\end{figure}

From the above lemma we clearly see that two states can be different, even very much, but still both offer huge usefulness for the port-based teleportation - please see Figure~\ref{fig:my_label}. Namely, from paper~\cite{ishizaka_quantum_2009} we know that fidelity of teleportation in non-optimal PBT, when  one uses $|\Phi^+\>_{AB}$, scales as $1-\mathcal{O}(1/N)$, while in OPBT, when one uses $|\Phi\>$, scales as $1-\mathcal{O}(1/N^2)$. However, considering optimal version of PBT not guarantee higher values of fidelity of the recycled state. In particular, it can be seen for $d=2$ where  considered values are even lower for optimal PBT than for non-optimal protocol.

\section{Discussion and further research}
\label{diss}
In this paper we start from analysis of the square-root measurements used in deterministic PBT from the point of view of their group-theoretical properties. We show how to evaluate matrix elements of measurements in every irreducible block and in what follows we established their composition rule. This led us to conclusion in which situation the considered measurements become projective. Next, applying these findings we found a way how to evaluate effectively square-roots from the measurements in PBT and their overlap with the signal states.

Having the knowledge about interior structure of the measurements we analyse the recycling protocol for deterministic non-optimal and optimal port-based teleportation for an arbitrary dimension of the port.  
By exploiting symmetries in PBT scheme we derive an explicit expression for $F(\mathcal{P}_{rec}(N,d,1))$ for an arbitrary dimension $d$ depending only on dimensions and multiplicities of the symmetric groups $S(N)$ and $S(N-1)$ in the Schur-Weyl duality. This expressions are effectively computable by using dedicated group theoretical packages (SAGE), together with created by the Authors codes in Python\cite{PanKoperRecycling}. In the particular case of qubits, where all irreps are indexed by Young frames with at most two rows, by using the Hook length formula, we derive closed expression for $F(\mathcal{P}_{rec}(N,2,1))$ depending only on the number of ports $N$. We also derived a lower bound on  quantity $F(\mathcal{P}_{rec}(N,d,k))$ which shows that the resource state is still useful, even after a few round of teleportation. Additionally, our results show that there is no explicit connection between the type the resource state in PBT and values of recycled fidelity. In particular, fidelity between the resource states in non-optimal and optimal schemes is low for large $N$, but the recycling fidelity does not differ much between them.

We have left also with a few open questions for further possible research. Th first of them is to consider what is the real fidelity of a teleported state after every round of the recycling protocol. From the paper~\cite{strelchuk_generalized_2013} and results presented here, we know that this fidelity is high and approaches 1 with $N\rightarrow \infty$, for fixed $k$. This is due to the observation that closeness of the idealised state $|\psi_{id}\>$ and the states $|\psi^{(a)}_{out}\>$ disturbed by Alice's measurement implies closeness of the resulting fideilites of teleportation. Having explicit values of fidelity for teleported state after every round of the recycling protocol would allow for direct comparison with other existing protocols for teleporting a number of states. One of such possibility is for example non- and optimal version of multi-PBT schemes introduced and investigated in a series of papers~\cite{strelchuk_generalized_2013,Kopszak2020,stud2020A,Mozrzymas2021optimalmultiport}. This however, requires proving and computing new properties of the recycling scheme which are out of scope of this manuscript. 

The second problem is to consider a recycling protocol for the probabilistic PBT scheme. In this particular case the situation is even more complicated comparing to its deterministic counterpart. This is due to the fact that we can consider two scenarios. In the first scenario we consider the resource state after success of transmission, while in the second one after failure of the whole procedure. We expect totally different behaviour in these two cases, since no transmission corresponds to a very destructive POVM applied by Alice. Also it is not clear how to consider more than a one round of the recycling scheme for probabilistic protocol - success or failure can happen interchangeably causing than final quality of the recycling scheme depends also on number of successes and failures. 
\section*{Acknowledgements}
MS, MM, PK are supported through grant Sonatina 2, UMO-2018/28/C/ST2/00004 from the Polish National Science Centre.

\appendix
\section{Proof of Lemma~\ref{thm1_rec}}
\label{AppA0}
The result presented here is analogous to the result from Section 2 of Supplementary Materials of~\cite{strelchuk_generalized_2013}, but it is necessary for understanding the discussion on the recycling protocol for optimal PBT. This is the reason why we have decided to include detailed reasoning in this manuscript in the notation used here and adapted for group-theoretical reasoning presented later.

The goal is  to calculate the expression
\be
\label{fideq}
F(\mathcal{P}_{rec}(N,d,1))=\sum_{a=1}^Np_aF(\psi_{out}^{(a)}, \psi_{id}^{(a)})=\frac{1}{d^{N+1}}\sum_{a=1}^N\operatorname{Tr}(\widetilde{\Pi}_a^{A_0A})F(\psi_{out}^{(a)}, \psi_{id}^{(a)}),
\ee
where $p_a= \frac{1}{d^{N+1}}\operatorname{Tr}(\widetilde{\Pi}_a^{A_0A})$. The operator $\psi_{id}^{(a)}=|\psi_{id}^{(a)}\>\<\psi_{id}^{(a)}|$ corresponds to the total state after the ideal process of teleportation, with the following explicit form
\begin{equation}
	\begin{split}
		|\psi^{(a)}_{id}\>&=|\psi^+_{A_0A_a}\>|\psi^+_{B_0B_a}\>\otimes \left(\bigotimes_{\substack{j=1\\j\neq a}}^N |\psi_{A_jB_j}^+\>\right)=\sqrt{d^{N+1}}\left(\sqrt{\sigma_{A_0A_a}}\otimes \mathbf{1}_{B_0B}\right)|\psi^+_{A_0B_0}\>\otimes \left(\bigotimes_{\substack{j=1}}^N |\psi_{A_jB_j}^+\>\right)\\
		&=\sqrt{d^{N+1}}\left(\sqrt{\sigma_{A_0A_a}}\otimes \mathbf{1}_{B_0B}\right)|\psi_{in}\>_{A_0B_0AB}.
	\end{split}
\end{equation}
In the above equation by $\sigma_{A_0A_a}=\frac{1}{d^{N-1}}\left(\mathbf{1}_{\overline{A}_0\overline{A}_a}\otimes P^+_{A_0A_a}\right)$ we denote the signal states, and by $\mathbf{1}_{\overline{A}_0\overline{A}_a}$ we denote identity operator acting on all systems $A_0A_1\cdots A_N$ but $A_0$ and $A_a$. The state $\psi_{out}^{(a)}=|\psi_{out}^{(a)}\>\<\psi_{out}^{(a)}|$ corresponds to the total state after application by Alice a measurement $\widetilde{\Pi}_a$ in non-idealised state, and it has a form
\begin{equation}
	\label{eq8}
	\begin{split}
		|\psi_{out}^{(a)}\>=\frac{\left(\sqrt{\widetilde{\Pi}_a^{A_0A}}\otimes \mathbf{1}_{B_0B}\right)|\psi_{in}\>_{A_0B_0AB}}{\left|\left|\left(\sqrt{\widetilde{\Pi}_a^{A_0A}}\otimes \mathbf{1}_{B_0B}\right)|\psi_{in}\>_{A_0B_0AB}\right|\right|_2}.
	\end{split}
\end{equation}
Let us calculate square of the norm from equation~\eqref{eq8}:
\begin{equation}
	\begin{split}
		&\left|\left|\left(\sqrt{\widetilde{\Pi}_a^{A_0A}}\otimes \mathbf{1}_{B_0B}\right)|\psi_{in}\>_{A_0B_0AB}\right|\right|_2^2=\<\psi_{in}|\widetilde{\Pi}_a^{A_0A}\otimes \mathbf{1}_{B_0B}|\psi_{in}\>=\tr\left[\left(\widetilde{\Pi}_a^{A_0A}\otimes \mathbf{1}_{B_0B}\right)|\psi_{in}\>\<\psi_{in}|\right]\\
		&=\tr\left[\left(\widetilde{\Pi}_a^{A_0A}\otimes \mathbf{1}_{B_0B}\right)\left(\psi^+_{A_0B_0}\otimes \left[\bigotimes_{\substack{j=1}}^N\psi_{A_jB_j}^+\right]\right)\right]=\frac{1}{d^{N+1}}\tr\left[\widetilde{\Pi}_a^{A_0A}\right],
	\end{split}
\end{equation}
so $|\psi_{out}^{(a)}\>$ reads
\begin{equation}
	|\psi_{out}^{(a)}\>=\sqrt{\frac{d^{N+1}}{\tr\left[\widetilde{\Pi}_a^{A_0A}\right]}}\left(\sqrt{\widetilde{\Pi}_a^{A_0A}}\otimes \mathbf{1}_{B_0B}\right)|\psi_{in}\>_{A_0B_0AB}.
\end{equation}
Now we are in position to calculate terms $F(\psi_{out}^{(a)}, \psi_{id}^{(a)})$ from~\eqref{fideq}. Since all the states $\psi_{out}^{(a)},\psi_{id}^{(a)}$ are pure, we have $F(\psi_{out}^{(a)}, \psi_{id}^{(a)})=|\braket{\psi_{out}^{(a)}|\psi_{id}^{(a)}}|=\left|\tr(|\psi^{(a)}_{id}\>\<\psi^{(a)}_{out}|)\right|.$
Due to permutational symmetry of the signals $\sigma_{A_0A_a}$ and measurements $\widetilde{\Pi}_a$ discussed in Section~\ref{dPBT}, without loss of generality we can compute $F(\psi_{out}^{(a)}, \psi_{id}^{(a)})$ only for $a=N$, this means that
\begin{equation}
	\label{F15}
	\begin{split}
		&F(\mathcal{P}_{rec})=\sum_{a=1}^Np_aF(\psi_{out}^{(a)}, \psi_{id}^{(a)})=\frac{N}{d^{N+1}}\tr(\widetilde{\Pi}_N^{A_0A})\left|\tr(|\psi^{(N)}_{id}\>\<\psi^{(N)}_{out}|)\right|\\
		&=\frac{N\tr(\widetilde{\Pi}_N^{A_0A})}{d^{N+1}}\sqrt{\frac{d^{N+1}}{\tr\left[\widetilde{\Pi}_N^{A_0A}\right]}}\sqrt{d^{N+1}}\left|\tr\left[\left(\sqrt{\widetilde{\Pi}_N^{A_0A}}\otimes \mathbf{1}_{B_0B}\right)|\psi_{in}\>\<\psi_{in}|\left(\sqrt{\sigma_{A_0A_N}}\otimes \mathbf{1}_{B_0B}\right)\right]\right|\\
		&=\frac{N\tr(\widetilde{\Pi}_N^{A_0A})}{d^{N+1}}\sqrt{\frac{d^{N+1}}{\tr\left(\widetilde{\Pi}_N^{A_0A}\right)}}\frac{\sqrt{d^{N+1}}}{d^{N+1}}\left|\tr\left(\sqrt{\widetilde{\Pi}_N^{A_0A}}\sqrt{\sigma_{A_0A_N}}\right)\right|\\
		&=\frac{N}{d^{N+1}}\sqrt{\tr\left(\widetilde{\Pi}_N^{A_0A}\right)}\left|\tr\left(\sqrt{\widetilde{\Pi}_N^{A_0A}}\sqrt{\sigma_{A_0A_N}}\right)\right|.
	\end{split}
\end{equation}
Now using relation
\begin{equation}
	\label{sqrt_sigma}
	\begin{split}
		\sqrt{\sigma_{A_0A_N}}&=\sqrt{\frac{1}{d^{N-1}}\left(\mathbf{1}_{\overline{A}_0\overline{A}_N}\otimes P^+_{A_0A_N}\right)}=\frac{1}{\sqrt{d^{N-1}}}\left(\mathbf{1}_{\overline{A}_0\overline{A}_N}\otimes P^+_{A_0A_N}\right)\\
		&=\frac{d^{N-1}}{\sqrt{d^{N-1}}}\left[\frac{1}{d^{N-1}}\left(\mathbf{1}_{\overline{A}_0\overline{A}_N}\otimes P^+_{A_0A_N}\right)\right]\\
		&=\sqrt{d^{N-1}}\sigma_{A_0A_N},
	\end{split}
\end{equation}
since $P^+_{A_0A_N}$ is a projector.
Inserting the above to~\eqref{F15} we obtain
\begin{equation}
	\begin{split}
		F(\mathcal{P}_{rec}(N,d,1))&=N\frac{\sqrt{d^{N-1}}}{d^{N+1}}\sqrt{\tr\left(\widetilde{\Pi}_N^{A_0A}\right)}\left|\tr\left(\sigma_{A_0A_N}\sqrt{\widetilde{\Pi}_N^{A_0A}}\right)\right|\\
		&=\frac{N}{d}\frac{\sqrt{\tr\left(\widetilde{\Pi}_N^{A_0A}\right)}}{\sqrt{d^{N+1}}}\left|\tr\left(\sigma_{A_0A_N}\sqrt{\widetilde{\Pi}_N^{A_0A}}\right)\right|.
	\end{split}
\end{equation}
To obtain the second equality from~\eqref{thmeq} it is enough to use the definition of the signal state $\sigma_{A_0A_N}$ and observe that
\begin{equation}
	\label{VV'}
	P^+_{A_0A_N}=\frac{1}{d}V'.
\end{equation}
This finishes the proof.

\section{Partially reduced irreducible representations}
\label{Prir}
The concept of partially reduced irreducible representations has been introduced in~\cite{Studzinski2017}, and its main goal is to simplify representation theoretic  calculations in the algebra $\A$. In our work, as we show later on, it plays central role in evaluation of explicit equations of square root from square-root measurements in deterministic port-based teleportation protocols. Here we remind only facts and ideas, which are necessary for potential reader of this manuscript. Most of the facts and definitions are taken from~\cite{Studzinski2017,MozJPA}

Let us consider an arbitrary unitary irrep $\psi^{\mu }$ of $S(n)$. It can be always unitarily transformed to reduced form
$\psi_{R}^{\mu }$, such that
\be
\label{eq:prir}
\forall \pi \in S(n-1)\quad \psi _{R}^{\mu }(\pi )=\bigoplus _{\alpha \in \mu
}\varphi ^{\alpha }(\pi),
\ee
where $\varphi ^{\alpha }$ are irreps of $S(n-1)$.  The sum runs over all Young frame $\alpha$ which can be obtained from a frame $\mu$ by subtracting a single box. We call decomposition given in~\eqref{eq:prir} the \textit{Partially Reduced Irreducible Representations (PRIR)}. We see that the restriction of the irrep  $\psi^{\mu }$ of $S(n)$ to the subgroup $S(n-1)$ has a block-diagonal form of completely reduced representation, which in matrix notation takes the form%
\be
\label{prir1}
\forall \pi \in S(n-1)\quad \psi _{R}^{\mu }(\pi)=\left( \delta ^{\alpha \beta
}\varphi _{i_{\alpha }j_{\alpha }}^{\alpha }\right) .
\ee
The block structure of this reduced representation allows us to introduce such a block indexation for PRIR $\psi _{R}^{\mu }$ of $S(n)$, which gives
\be
\forall \sigma \in S(n)\quad \psi _{R}^{\mu }(\sigma )=\left( \psi _{i_{\alpha
	}j_{\beta }}^{\alpha \beta }(\sigma )\right) ,
\ee
where the matrices on the diagonal $(\psi _{R}^{\mu })^{\alpha \alpha
		}(\sigma )=\left( \psi _{i_{\alpha }j_{\alpha }}^{\alpha \alpha }(\sigma )\right) $ are
of dimension of corresponding irrep $\varphi ^{\alpha }$ of $S(n-1)$. The
off diagonal blocks need not to be square. 
The PRIR notation allows us for relative friendly description of the basic objects in the algebra $\A$, being also a building blocks in deterministic PBT scheme. In particular, we have
\begin{proposition}[extended version of Prop. 33, see page 14 of~\cite{MozJPA}]
	\label{BP20}
	In the irrep $\Phi^{\alpha }$ of the algebra $%
		\mathcal{A}_{n}^{t_{n}}(d)$ we have the following matrix representation of elements $%
		V'[(a,n)]$
	\be
	\label{blee1}
	M_{f}^{\alpha }\left[V'[(a,n)]\right]_{j_{\xi_{\omega}} \ j_{\xi_{\nu}}}^{\xi_{\omega} \ \xi_{\nu}}=\frac{1}{n-1}\frac{\sqrt{d_{\xi}d_{\omega}}}{d_{\alpha}}\sum_{k_{\alpha}}\sqrt{\gamma_{\omega}(\alpha)}\psi_{R \ j_{\xi_{\omega}} \ k_{\alpha}}^{\omega \ \xi_{\omega} \ \alpha}[(a, n-1)]\psi_{R \ k_{\alpha} \ j_{\xi_{\nu}}}^{\nu \
		\alpha \ \xi_{\nu}}[(a, n-1)]\sqrt{\gamma_{\nu}(\alpha)},
	\ee
	where $\omega ,\nu \neq \theta$ and the subscript $f$  means that the matrix representation is calculated
	in reduced basis $%
		f\equiv\{f_{j_{\nu} }^{\nu }:h(\nu )\leq d,\quad j_{\nu} =1,\ldots,d_{\nu }\}$ of the ideal $\Phi^{\alpha }$.

	In particular for $a=n-1$ expression~\eqref{blee1} reduces to
	\be
	\label{blee2}
	M_f^{\alpha}\left(V' \right)_{j_{\xi_{\omega}} \ j_{\xi_{\nu}}}^{\xi_{\omega} \ \xi_{\nu}}=\frac{1}{n-1}\frac{\sqrt{d_{\xi}d_{\omega}}}{d_{\alpha}}\sqrt{\gamma_{\omega}(\alpha)\gamma_{\nu}(\alpha)}\delta^{\xi_{\omega}\alpha}\delta^{\xi_{\nu}\alpha}\delta_{j_{\xi_{\omega} j_{\xi_{\nu}}}}.
	\ee
	In the above expressions numbers $\gamma_{\nu}(\alpha)$ equal to
	\be
	\gamma_{\nu}(\alpha)=(n-1)\frac{m_{\nu}d_{\alpha}}{m_{\alpha}d_{\nu}}
	\ee
	are eigenvalues of the sum the following operator
	\be
	\widetilde{\rho}=\sum_{i=1}^N V'[(a,n)].
	\ee
\end{proposition}
The last sentence from the above Proposition is not obvious and it has been proven in Proposition 2 of~\cite{Studzinski2017}. In fact this proposition states that the non-zero eigenvalues of the operator $\rho$ given in~\eqref{eq:rhoV} are of the form
\begin{equation}
	\label{llambda}
	\lambda_{\nu}(\alpha)=\frac{N}{d^{N}}\frac{m_{\nu}d_{\alpha}}{m_{\alpha}d_{\nu}},\quad \text{where}\quad N=n-1.
\end{equation}
We can say even more (see Theorem 1 in~\cite{Studzinski2017}), namely the operator $\rho$ admits the following spectral decomposition
\begin{equation}
	\label{rho_spectral}
	\rho=\frac{1}{d^N}\widetilde{\rho}=\sum_{\alpha \vdash N-1}\sum_{\nu \in \alpha}\lambda_{\nu}(\alpha)F_{\nu}(\alpha),
\end{equation}
where $F_{\nu}(\alpha)$ are projectors on irreps of $\A$ described briefly in Section~\ref{tools}.

Since the projectors $F_{\nu}(\alpha)$ play a central role in our considerations we also need an explicit form of operators $F_{\mu}(\alpha)$ in PRIR representation:
\begin{lemma}[Lemma 35, page 15 of~\cite{MozJPA}]
	\label{lemma_Mf}
	The matrix form of the projector $F_{\nu }(\alpha )$ on non-trivial irreducible spaces of the algebra $\mathcal{A}_{n}^{t_{n}}(d)$, in the reduced basis $f$ has
	the following form
	\be
	M_{f}^{\alpha }[F_{\nu }(\alpha )]_{\xi _{\eta }j_{\xi _{\eta }}\xi _{\mu
	}j_{\xi _{\mu }}}^{\eta \qquad \mu }=\delta ^{\eta \nu }\delta ^{\nu \mu
	}\delta _{\xi _{\eta }\xi _{\mu }}\delta _{j_{\xi _{\eta }}j_{\xi _{\mu }}},
	\ee
	i.e.  in the
	reduced basis $%
		f\equiv\{f_{j_{\nu} }^{\nu }:h(\nu )\leq d,\quad j_{\nu} =1,\ldots,d_{\nu }\}$ of the ideal $\Phi^{\alpha }$, the projector $F_{\nu }(\alpha )$ takes its canonical
	form with one$'$s on the diagonal in the position of the irrep $%
		\psi ^{\nu }$ of the group $S(n-1)$ only.
\end{lemma}

\begin{corollary}[Corollary 36, page 15 of~\cite{MozJPA}]
	\label{CM4}
	\be
	\tr M_{f}^{\alpha }[F_{\nu }(\alpha )]=d_{\nu },
	\ee
	and from this we get
	\be
	\label{FF2}
	\tr_{\mathcal{H}}F_{\nu }(\alpha )=m_{\alpha }d_{\nu },
	\ee
	where $\mathcal{H}=(\mathbb{C}^d)^{\ot n}$, and $m_{\alpha}$ is the multiplicity  the irreps $\varphi ^{\alpha }$ of $S(n-2)$ in the representation $V(S(n-2))$.
\end{corollary}

\section{Additional properties of the optimising operation $O_A$}
In optimal PBT from paper~\cite{StuNJP} we know that Alice to increase efficiency of the protocol has to apply to her part of shared maximally entangled pairs operation $O_A$ of the form
\begin{equation}
	\label{Oa}
	O_A=\sqrt{d^N}\sum_{\mu \vdash N} \frac{v_{\mu}}{\sqrt{d_{\mu}m_{\mu}}}P_{\mu},
\end{equation}
where the non-negative coefficients $v_{\mu}$ are entries of the eigenvector corresponding to the maximal eigenvalue of teleportation matrix discussed in Section 4 of the same work~\cite{StuNJP}. Now, we prove the following
\begin{fact}
	\label{trOXO}
	For every $1\leq a \leq N$, there is
	\begin{equation}
		\tr\left(O_A^{\dagger}\widetilde{\Pi}_a O_A\right)=\tr\left(\widetilde{\Pi}_a\right)=\frac{d^{N+1}}{N},
	\end{equation}
	where $\widetilde{\Pi}_a$ are POVMs with an additional part $\Delta$ as it is described in expression~\eqref{mea2}.
\end{fact}

\begin{proof}
	First let us observe that due to form of $O_A$ given in~\eqref{Oa} it commutes with all permutations from $S(N)$. On the other hand from~\eqref{mes_cov}, we know that $\widetilde{\Pi}_a$ are covariant with respect to the elements from $S(N)$. These properties allow us to write
	\begin{equation}
		\begin{split}
			&\sum_{a=1}^N\tr\left(O_A^{\dagger}\widetilde{\Pi}_aO_A\right)=\tr\left(O_A^{\dagger}\mathbf{1}_{(\mathbb{C}^d)^{\otimes n}}O_A\right)=d\tr\left(O_A^{\dagger}O_A\right)=d^{N+1},\\
			&N\tr\left(O_A^{\dagger}\widetilde{\Pi}_aO_A\right)=d^{N+1}.
		\end{split}
	\end{equation}
	In the first equality we use that operators $\widetilde{\Pi}_a$ are POVMs and they have to sum up to identity on the whole $(\mathbb{C}^d)^{\otimes n}$ space. In the second equality we use fact that $O_A$ acts on $N=n-1$ first systems. The third equality is due to normalisation condition $\tr(O_A^{\dagger}O_A)=d^N$. Finally, to get the second line we use mentioned covariance of POVMs.
\end{proof}

\section{Calculations for the recycling protocol for optimal deterministic PBT}
\label{rec_OdPBT}
In the optimal deterministic port-based teleportation (OdPBT), as we described earlier, Alice optimises over the shared maximally entangled pairs and the measurements before she runs the protocol. This optimisation results in application of the global operation $O_A$ on her halves of entangled states, see equation~\eqref{resource}. The goal of this section is to re-derive Theorem~\ref{thm1_rec} for the optimal protocol. We start from definitions of the ideal and the real state after the teleportation process.

The ideal state $|\psi_{id}^{(i)}\>$ after teleportation process is given as
\begin{align}
	|\psi_{id}^{(i)}\>= \ket{\psi^+}_{A_0A_i}\otimes\ket{\psi^+}_{B_0B_i}\otimes \ket{\psi}_{\overline{A}_i\overline{B}_i},
\end{align}
where $\overline{A}_i\overline{B}_i$ denotes all subsystems except this on $i$-th position.
After the ideal process of teleportation we would like the parties to share also ideal resource state, except ideally teleported systems. It would mean that the state $\ket{\psi}_{\overline{A}_i\overline{B}_i}$ should be optimal for OdPBT performed on $N-1$ ports:
\begin{align}
	\ket{\psi}_{\overline{A}_i\overline{B}_i}=\left(O_{\widetilde{A}}\otimes \mathbf{1}\right)\bigotimes_{j=1,j\neq i}^{N} \ket{\psi^+}_{A_jB_j},
\end{align}
where $\widetilde{A}=A_1A_2\cdots \overline{A}_i\cdots A_N$.
This leads us to
\begin{align}
	|\psi_{id}^{(i)}\> & = \ket{\psi^+}_{A_0A_i}\otimes\ket{\psi^+}_{B_0B_i}\otimes \left(O_{\widetilde{A}}\otimes \mathbf{1}\right)\bigotimes_{j=1,j\neq i}^{N} \ket{\psi^+}_{A_jB_j}                                                  \\
	                   & =\sqrt{d^{N+1}}\sqrt{\sigma_{A_0A_i}}\left(\ket{\psi^+}_{A_0B_0}\otimes\ket{\psi^+}_{A_iB_i}\right)\otimes \left(O_{\widetilde{A}}\otimes \mathbf{1}\right)\bigotimes_{j=1,j\neq i}^{N} \ket{\psi^+}_{A_jB_j}.
\end{align}
As it was shown in~\cite{leditzky2020optimality} the optimal measurements for OdPBT coincide with those for non-optimal PBT, but instead distinguishing signals $\{\sigma_{A_0A_i}\}_{i=1}^N$ we distinguish their rotated versions $\{O_A\sigma_{A_0A_i}O_A^{\dagger}\}_{i=1}^N$. It means we can use measurements from~\eqref{mea2} and the total state after application of a measurement $\widetilde{\Pi}_i^{A_0A}$, acting non-trivially on systems $A_0A$, with $A=A_1A_2\cdots A_N$, equals to
\begin{align}
	\label{Onorm}
	|\psi^{(i)}_{out}\>=\frac{\left(\sqrt{\widetilde{\Pi}_i^{A_0A}}\otimes \mathbf{1}\right)\left(|\psi^+\>_{A_0B_0}\otimes \left(O_A\otimes \mathbf{1}\right)\bigotimes_{j=1}^{N} \ket{\psi^+}_{A_jB_j}\right)}{\left|\left|\left(\sqrt{\widetilde{\Pi}_i^{A_0A}}\otimes \mathbf{1}\right)\left(|\psi^+\>_{A_0B_0}\otimes \left(O_A\otimes \mathbf{1}\right)\bigotimes_{j=1}^{N} \ket{\psi^+}_{A_jB_j}\right)\right|\right|_2}.
\end{align}
Having definitions of states $|\psi_{id}^{(i)}\>$ and $|\psi^{(i)}_{out}\>$ in OdPBT we are in position  to re-formulate Theorem~\ref{thm1_rec} from the main  text.
\begin{theorem}
	\label{thm:fidel}
	The fidelity $F(\mathcal{P}_{rec}(N,d,1))$ in the recycling scheme for the OdPBT scheme, with $N$ ports, each of dimension $d$, after one round of teleportation is the following:
	\begin{equation}
		\label{fidel}
		\begin{split}
			F(\mathcal{P}_{rec}(N,d,1))=
			\frac{\sqrt{N}}{d}\left|\tr\left(\sigma_{A_0A_N}\sqrt{\widetilde{\Pi}_N^{A_0A}}O_AO_{\widetilde{A}}^{\dagger}\right)\right|,
		\end{split}
	\end{equation}
	where $\sigma_{A_0A_N},\widetilde{\Pi}_N^{A_0A}$ are respectively the signal state and the measurement corresponding to index $a=N$ in~\eqref{eq:measurements}. Operators $O_A,O_{\widetilde{A}}$ are operations applied by Alice on her halves of shared maximally entangled state to increase the efficiency of the protocol, respectively for $N$ and $N-1$ ports.
\end{theorem}

\begin{proof}
	We start from computing norm in equation~\eqref{Onorm}. One can show that we have
	\begin{align}
		\left|\left|\left(\sqrt{\widetilde{\Pi}_i^{A_0A}}\otimes
		\mathbf{1}\right)\left(|\psi^+\>_{A_0B_0}\otimes \left(O_A\otimes \mathbf{1}\right)\bigotimes_{j=1}^{N} \ket{\psi^+}_{A_jB_j}\right)\right|\right|_2^2=\frac{1}{d^{N+1}}\tr\left(O_A^{\dagger}\widetilde{\Pi}_i^{A_0A}O_A\right).
	\end{align}
	In the next step we evaluate fidelity $F(|\psi_{id}^{(i)}\>,|\psi^{(i)}_{out}\>)=\left|\tr\left(|\psi^{(i)}_{out}\>\<\psi_{id}^{(i)}|\right)\right|$ between ideal and the real situation:
	\begin{equation}
		\begin{split}
			&\left|\tr\left(|\psi^{(i)}_{out}\>\<\psi_{id}^{(i)}|\right)\right|=\frac{d^{N+1}}{\sqrt{\tr\left(O_A^{\dagger}\widetilde{\Pi}_i^{A_0A}O_A\right)}}\times\\
			&\times\left|\tr\left(\left(\sqrt{\widetilde{\Pi}_i^{A_0A}}\otimes \mathbf{1}\right)\left(|\psi^+\>_{A_0B_0}\otimes \left(O_A\otimes \mathbf{1}\right)\bigotimes_{j=1}^{N} \ket{\psi^+}_{A_jB_j}\right)\left(\<\psi^+|_{A_0B_0}\otimes\<\psi^+|_{A_iB_i}\right)\sqrt{\sigma_{A_0A_i}}\otimes \bigotimes_{j=1,j\neq i}^{N} \<\psi^+|_{A_jB_j}\left(O_{\widetilde{A}}^{\dagger}\otimes \mathbf{1}\right)\right)\right|\\
			&=\frac{1}{\sqrt{\tr\left(O_A^{\dagger}\widetilde{\Pi}_i^{A_0A}O_A\right)}}\left|\tr\left(\sqrt{\widetilde{\Pi}_i^{A_0A}}\left(O_A\otimes \mathbf{1}_{A_0}\right)\sqrt{\sigma_{A_0A_i}}\left(O_{\widetilde{A}}^{\dagger}\otimes \mathbf{1}_{A_0A_i}\right)\right)\right|\\
			&=\frac{\sqrt{d^{N-1}}}{\sqrt{\tr\left(O_A^{\dagger}\widetilde{\Pi}_i^{A_0A}O_A\right)}}\left|\tr\left(\sqrt{\widetilde{\Pi}_i^{A_0A}}\left(O_A\otimes \mathbf{1}_{A_0}\right)\sigma_{A_0A_i}\left(O_{\widetilde{A}}^{\dagger}\otimes \mathbf{1}_{A_0A_i}\right)\right)\right|,
		\end{split}
	\end{equation}
	where to obtain the last line we use property from equation~\eqref{sqrt_sigma} from the main text.
	As it was discussed in Section~\ref{dPBT} the measurements and the signals are  covariant with respect to permutations $V(a,N)$, for $a=1,\ldots,N$. Next, due to definition of $O_A$ given in~\eqref{Oa} we see that  it is enough to calculate the above expression for $i=N$, so we have
	\begin{align}
		\left|\tr\left(|\psi^{(N)}_{out}\>\<\psi_{id}^{(N)}|\right)\right|=\frac{\sqrt{d^{N-1}}}{\sqrt{\tr\left(O_A^{\dagger}\widetilde{\Pi}_N^{A_0A}O_A\right)}}\left|\tr\left(\sqrt{\widetilde{\Pi}_N^{A_0A}}O_A\sigma_{A_0A_N}O_{\widetilde{A}}^{\dagger}\right)\right|,
	\end{align}
	where we suppressed the identity operators to simplify the notation.
	Then the fidelity$F(\mathcal{P}_{rec})$, due to~\eqref{PF} reads:
	\begin{equation}
		\begin{split}
			F(\mathcal{P}_{rec}(N,d,1))&=\sum_{i=1}^Np_iF(|\psi_{id}^{(i)}\>,|\psi^{(i)}_{out}\>)=Np_NF(|\psi_{id}^{(N)}\>,|\psi^{(N)}_{out}\>)\\
			&=Np_N\frac{\sqrt{d^{N-1}}}{\sqrt{\tr\left(O_A^{\dagger}\widetilde{\Pi}_N^{A_0A}O_A\right)}}\left|\tr\left(\sqrt{\widetilde{\Pi}_N^{A_0A}}O_A\sigma_{A_0A_N}O_{\widetilde{A}}^{\dagger}\right)\right|\\
			&=\frac{N\tr(\widetilde{\Pi}_N^{A_0A})}{d^{N+1}}\frac{\sqrt{d^{N-1}}}{\sqrt{\tr\left(O_A^{\dagger}\widetilde{\Pi}_N^{A_0A}O_A\right)}}\left|\tr\left(\sqrt{\widetilde{\Pi}_N^{A_0A}}O_A\sigma_{A_0A_N}O_{\widetilde{A}}^{\dagger}\right)\right|\\
			&=\frac{N\tr(\widetilde{\Pi}_N^{A_0A})}{d\sqrt{d^{N+1}}\sqrt{\tr\left(O_A^{\dagger}\widetilde{\Pi}_N^{A_0A}O_A\right)}}\left|\tr\left(\sigma_{A_0A_N}\sqrt{\widetilde{\Pi}_N^{A_0A}}O_AO_{\widetilde{A}}^{\dagger}\right)\right|,
		\end{split}
	\end{equation}
	since $p_i=\tr(\widetilde{\Pi}_i^{A_0A})/d^{N+1}$ and $[\sigma_{A_0A_N},O_{\widetilde{A}}^{\dagger}]=0$. Finally, applying Fact~\ref{trOXO} to the denominator of the above expression we obtain the first line from~\eqref{fidel}. To get the second expression from~\eqref{fidel} we have to reasoning from~\ref{trPiN}. This completes the proof.
\end{proof}

Please notice that plugging $O_A=\mathbf{1}_{(\mathbb{C}^d)^{\ot N}}$ and $O_A=\mathbf{1}_{(\mathbb{C}^d)^{\ot {N-1}}}$, we reduce to the statement of Theorem~\ref{thm1_rec} from the main text corresponding to the non-optimal deterministic PBT.  Having the general expression for $F(\mathcal{P}_{rec}(N,d,1))$ in~\eqref{fidel} in terms of operators describing optimal procedure we are ready to formulate theorem connecting the efficiency of the recycling protocol with group theoretic quantities as it was for non-optimal scheme in Theorem~\ref{thm1_rec}. First we prove the following technical proposition:

\begin{proposition}
	\label{Prop_aux1}
	Let $\mu \vdash N$ and $\alpha \vdash N-1$ label irreps of $S(N)$ and $S(N-1)$ respectively, then the following relation holds:
	\begin{equation}
		\tr\left[P_{\mu }P_{\alpha }V'\sqrt{\widetilde{\Pi}^{A_0A}_N}\right]=\delta _{\alpha ,\mu -\square }\frac{d_{\alpha }\sum_{\substack{\nu \neq \theta\\ \nu=\alpha+\square }}%
		\sqrt{m_{\nu }d_{\nu }}}{\sqrt{Nd_{\alpha}-d_{\theta }}}\frac{%
		\sqrt{m_{\mu }d_{\mu }}}{\sqrt{Nd_{\alpha }}},
	\end{equation}%
	where $P_{\alpha},P_{\mu}$ denote Young projectors, operator $V'$ is given through~\eqref{VV'} and measurement $\widetilde{\Pi}^{A_0A}_N$ in~\eqref{mea2}.
	The symbol $\delta _{\alpha ,\mu -\square }$ means that if a Young frame $\alpha \vdash N-1$ is not related to $\mu \vdash N$ by adding a single box then $\delta _{\alpha ,\mu -\square }=0$ and the resulting trace is zero, otherwise $\delta _{\alpha ,\mu -\square }=1$.
\end{proposition}

\begin{proof}
	The calculation of the trace $\tr\left[P_{\mu }P_{\alpha }V'\sqrt{\widetilde{\Pi}^{A_0A}_N}\right]$ is based on the decomposition of natural
	representation of the algebra of partially transposed operators $%
		\mathcal{A}_{n}^{t_{n}}(d)$ with carrier space $\mathcal{H}=(
		\mathbb{C}^{d})^{\otimes n}$ onto  irreducible representations $\ M_f^{\alpha }$ of $%
		A_{n}^{t_{n}}(d)$, where $\alpha \vdash N-1$ labels the irreps of the algebra $\mathcal{A}_{n}^{t_{n}}(d)$, see Appendix~\ref{Prir}. Then we calculate the
	corresponding matrices $\ M_f^{\alpha }(P_{\mu })$, $M_f^{\alpha }(P_{\beta }),$ $%
		\ M_f^{\alpha }(V')\ $and $M_f^{\alpha }(\sqrt{\widetilde{\Pi}^{A_0A}_N})$,
	where in order to calculate the last case, we use spectral decomposition of
	the operator $\widetilde{\Pi}^{A_0A}_N$ . Next, we derive the matrix $M_f^{\alpha }(P_{\mu
			}P_{\beta }V'\sqrt{\widetilde{\Pi}^{A_0A}_N})$ for each irrep $\alpha$ and calculate its trace. The final formula for the trace
	is
	\begin{equation}
		\tr\left[P_{\mu }P_{\beta }V'\sqrt{\widetilde{\Pi}^{A_0A}_N})
		\right]=\sum_{\alpha \vdash N-1}m_{\alpha }\tr\left[M_f^{\alpha }(P_{\mu
			}P_{\beta }V'\sqrt{\widetilde{\Pi}^{A_0A}_N})\right],
	\end{equation}
	where $m_{\alpha }$ is the multiplicity of the irrep $M_f^{\alpha }$ in the
	natural representation of the algebra  $\mathcal{A}_{n}^{t_{n}}(d).$ The rest
	the proof is analogous to calculations in Appendix~\ref{StructurePOVMs}, and we leave it for the reader.
\end{proof}

Having Theorem~\ref{thm:fidel} and Proposition~\ref{Prop_aux1} from this appendix we can present the proof of Theorem~\ref{expOA} from the main text:
\begin{proof}[Proof of Theorem~\ref{expOA}]
	We prove the statement by the straightforward calculations
	\begin{equation}
		\begin{split}
			F(\mathcal{P}_{rec}(N,d,1))&=\frac{\sqrt{N}}{d}\left|\tr\left(\sigma_{A_0A_N}\sqrt{\widetilde{\Pi}_N^{A_0A}}O_AO_{\widetilde{A}}^{\dagger}\right)\right|=\frac{\sqrt{N}}{d}\left|\tr\left(\sigma_{A_0A_N}\sqrt{\Pi_N^{A_0A}}O_AO_{\widetilde{A}}^{\dagger}\right)\right|\\
			&=\frac{\sqrt{N}}{d\sqrt{d}}\sum_{\alpha \vdash N-1}\sum_{\mu \vdash N} \frac{v_{\alpha}v_{\mu}}{\sqrt{d_{\mu}m_{\mu}d_{\alpha}m_{\alpha}}}\tr\left(V'\sqrt{\widetilde{\Pi}^{A_0A}_N}P_{\mu}P_{\alpha}\right).
		\end{split}
	\end{equation}
	The second equality follows from the fact the the extra term $\frac{1}{N}\Delta$ in definition of measurements $\widetilde{\Pi}_a^{A_0A}$ in~\eqref{mea2} is always orthogonal to $\sigma_{A_0A_a}$, for $1\leq a\leq N$, so it is enough tho work only with the term $\Pi_a^{A_0A}$. To get the third equality we use fact that $\sigma_{A_0A_a}=\frac{1}{d^N}V'_{A_0A_N}\equiv \frac{1}{d^N}V'$ and by plugging the explicit forms of operators $O_A, O_{\widetilde{A}}$ given in~\eqref{Oa} to expression~\eqref{fidel} in Theorem~\ref{thm:fidel}. Now using the statement of Proposition~\ref{Prop_aux1} we finish the proof.
\end{proof}

\section{Teleportation matrix in the qubit case}
\label{AppE}
The teleportation matrix $M_F$ has been introduce firstly in~\cite{StuNJP} and its maximal eigenvalue $\lambda_{\max}(M_F)$ encodes the entanglement fidelity $F$ in optimised deterministic PBT:
\begin{equation}
	F=\frac{1}{d^2}\lambda_{\max}(M_F).
\end{equation}
From the considerations in this paper and in~\cite{StuNJP} we know that the optimising operation $O_A$ from~\eqref{Oa} can be expressed by entries of eigenvector $v=(v_{\mu})$ corresponding to the maximal eigenvalue $\lambda_{\max}(M_F)$. However, analytical expressions for eigenvectors and eigenvalues of $M_F$ are known only in two cases, when $d\geq N$ and $d=2$ with arbitrary $N$ (see Section 4 and Section 5.3 in~\cite{StuNJP}). In the latter case the interior structure of $M_F$ is reasonably simple and the whole matrix is a tridiagonal matrix of the form:
\be
M_F=\frac{1}{4}\begin{pmatrix}
	-x_1+2 & 1      & 0      & 0      & \cdots & 0      & 0      \\\
	1      & 2      & 1      & 0      & \cdots & 0      & 0      \\\
	0      & 1      & 2      & 1      & \cdots & 0      & 0      \\
	\vdots & \vdots & \vdots & \vdots & \ddots & \vdots & \vdots \\
	0      & 0      & 0      & 0      & \cdots & 2      & 1      \\
	0      & 0      & 0      & 0      & \cdots & 1      & -x_2+2
\end{pmatrix} \in \operatorname{M}(t,\mathbb{R}),
\ee
where
\be
t=\floor{N/2+1}=\begin{cases}
	\frac{N}{2}+1\quad N - \text{even}, \\
	\frac{N+1}{2}\quad N- \text{odd}.
\end{cases}
\ee
The values $x_1,x_2$ also depend on the parity of $N$, and we have
\be
\begin{cases}
	x_1=x_2=1\quad N - \text{even}, \\
	x_1=1,x_2=0\quad N- \text{odd}.
\end{cases}
\ee
The index $t$ numerates all irreps in the qubit case, so every number $t$ corresponds to some Young frame $\mu$ with up to two rows and $N$ boxes of the form $\mu=(N-l,l)$. In this particular qubit case we can exploit results of  Losonczi from \cite{Losonczi1992}, where direct expressions for eigensystem are given. By exploiting his results directly one can find the following expressions for the entries of the vector $v=(v^{(N)}_l)$ corresponding to maximal eigenvalue of $M_F$:
\be
\label{eigend=2}
v_l^{(N)}=\begin{cases} (-1)^{\frac{N}{2}-l}
		{\left(\operatorname{sin}
	\frac{\left(\frac{N+2}{2}-l\right) N\pi}{N+2} - \operatorname{sin}\frac{\left(\frac{N}{2}-l\right) N\pi}{N+2}\right)}\bigg/{\operatorname{sin}\frac{ N\pi}{N+2}}\; & \textit{for even N} \\
	(-1)^{\frac{N-1}{2}-l}\operatorname{sin} \frac{\left(\frac{N+1}{2}-l\right) N\pi}{N+2}\big/{\operatorname{sin}\frac{ N\pi}{N+2}}\;                                 & \textit{for odd N.}
\end{cases}
\ee
For further reasons the vectors $v=(v^{(N)}_l)$ have to be normalised, however for transparency we do not introduce here a new notation for their normalised versions and we use the same symbol everywhere in the text.

\section{Proof of Lemma~\ref{lem:f_opt}}
\label{app:opbtd2}
In this section we derive the expression for resource state fidelity in optimal recycling procedure, in qubit case. The general expression reads
\be
\begin{split}
	F(\mathcal{P}_{rec}(N,d,1))&=\frac{1}{d^{1/2}}\sum_{\alpha \vdash N-1}\sum_{\mu \in \alpha}\frac{v_{\alpha}v_{\mu}}{m_{\alpha}^{1/2}}\frac{\sum_{\substack{\nu \neq \theta\\
			\nu\in \alpha}}\sqrt{m_{\nu}d_{\nu}}}{\sqrt{Nd_{\alpha}-d_{\theta}}},
\end{split}
\ee
Using the expressions for dimensionality $d_\alpha$~given by \eqref{eq:m_d_alpha} and multiplicity $m_\alpha$ given by \eqref{eq:d_theta} together with the formula for the components of the normalised eigenvector 
$v_l^{(N)}$ and given in~\eqref{eigend=2}
and summing over possible lengths of second row of Young Tableaux corresponding to a given irrep $\alpha=(N-l,l)$ we obtain the final expression in the qubit case
\be
\begin{split}
	&F(\mathcal{P}_{rec}(N,2,1)) =\\
	&=\frac{1}{\sqrt 2}\sum _{l=0}^{k }
	\frac{v_l^{(N-1)}\left({v_l^{(N)} + v_{l+1}^{(N)}}\right)}{N-2l}
	\left ((N-2l+1)\sqrt{\frac{{N+1 \choose l}}{N+1}} + (N-2l-1)\sqrt{\frac{{N+1 \choose l+1}}{N+1}} \right)^2
	\sqrt{\frac{(N-l+1)(l+1)}{(N-2l)(N+1){N\choose l}}}.
\end{split}
\ee
In the above expression we use the fact, that in qubit case there are only two possibilities of adding a single box to $\alpha$, which is dented by $\mu\in \alpha$, resulting in two irreps $(N+1-l, l)$ and $(N-l, l+1)$, the latter of which is valid only when $l+1\leq{N}/{2}$, otherwise it is set to 0.
\section{Proof of Lemma~\ref{l:FPBT} with equivalent quantum angular momentum picture}
\label{AppB}

First let us derive formula for the fidelity $F(|\Phi^+\>_{AB},|\Phi\>_{AB})$ for an arbitrary dimension $d$ using explicit form of the operation $O_A$ given in~\eqref{expOA}:
\begin{equation}
	\label{compa}
	\begin{split}
		F(|\Phi^+\>_{AB},|\Phi\>_{AB})&=|\tr|\Phi^+\>\<\Phi|_{AB}|=|\tr\left( (O_A \otimes \mathbf{1}_B)|\Phi^+\>\<\Phi^+|_{AB}\right)|\\
		&=\frac{1}{\sqrt{d^N}}\tr(O_A)=\frac{1}{\sqrt{d^N}}\sum_{\mu \vdash N}v_{\mu}\sqrt{d_{\mu}m_{\mu}},
	\end{split}
\end{equation}
since $\tr(P_{\mu})=d_{\mu}m_{\mu}$ and $\tr_B(|\Phi^+\>\<\Phi^+|_{AB})=(1/d^N)\mathbf{1}_A$. As it was said earlier coefficients $v_{\mu}\geq 0$ are entries of the eigenvector corresponding to maximal eigenvalue of the teleportation matrix $M_F$ introduced in~\cite{StuNJP}. For $d>2$ the coefficients $v_{\mu}$ can be computed only using numerical methods. For $d=2$ we have two options. The first one is to observe that in this case the matrix $M_F$ is tri-diagonal, for which analytical expressions for eigenvalues and eigenvectors are known due to Losonczi work~\cite{Losonczi1992} and are given by~\eqref{eq:v_l}. Using the expressions provided in~\eqref{eq:m_d_alpha} and summing over all possible lenghts of lower row in Young tableaux $(N-l, l)$, we have the following formula
\be
\begin{split}
	&F(\mathcal{P}_{rec}(N,2,1)) =\\
	&=\frac{1}{\sqrt 2}\sum _{l=0}^{k }
	\frac{v_l^{(N-1)}\left({v_l^{(N)} + v_{l+1}^{(N)}}\right)}{N-2l}
	\left ((N-2l+1)\sqrt{\frac{{N+1 \choose l}}{N+1}} + (N-2l-1)\sqrt{\frac{{N+1 \choose l+1}}{N+1}} \right)^2
	\sqrt{\frac{(N-l+1)(l+1)}{(N-2l)(N+1){N\choose l}}}.
\end{split}
\ee

However, for our purposes it is enough to use straightforwardly results contained in the seminal work of Hiroshima and Ishizaka~\cite{ishizaka_quantum_2009}. They have described PBT protocols using tools coming from representation theory of $SU(2)^{\otimes N}$ group. Therefore we use a representation in the spin angular momentum for the $N-$spin system. In this representation the operator $O_A^{\dagger}O_A$ reads as
\begin{equation}
	\label{XA}
	O_A^{\dagger}O_A=O_A^2=\sum_{j=j_{\min}}^{N/2}\gamma(j)\mathds{1}(j),\qquad \gamma(j)\geq 0.
\end{equation}
The sum in~\eqref{XA} runs from $j_{\min}=0(1/2)$ for $N$ even (odd). The operator $\mathds{1}(j)$ is the identity operator for a fixed quantum number $j$. The operator $\mathds{1}(j)$ corresponds directly to a Young projector $P_{\mu}$ in expression~\eqref{compa}. This, together with explicit form of the coefficients $\gamma(j)$ coming from the optimisation in OPBT (see~\cite{ishizaka_quantum_2009}), allows us to write
\begin{equation}
	\begin{split}
		O_A=\sum_{j=j_{\min}}^{N/2}\sqrt{\gamma(j)}\mathds{1}(j)&=\sum_{j=j_{\min}}^{N/2}\sqrt{\frac{2^{N+2}}{(N+2)(2j+1)d_j}\operatorname{sin}^2\left(\frac{\pi(2j+1)}{N+2}\right)}\mathds{1}(j)\\
		&=\sum_{j=j_{\min}}^{N/2}\operatorname{sin}\left(\frac{\pi(2j+1)}{N+2}\right)\sqrt{\frac{2^{N+2}}{(N+2)(2j+1)d_j}}\mathds{1}(j),
	\end{split}
\end{equation}
since for $j_{\min}\leq j\leq N/2$ the sine function gives always positive values.
Taking trace from the above expression and taking into account that $\tr \mathds{1}(j)=d_jm_j $ and
\begin{equation}
	d_j=\frac{(2j+1)N!}{(N/2-j)!(N/2+j+1)!},\qquad m_j=2j+1,
\end{equation}
we obtain our result of the form
\begin{equation}
	F(|\Phi^+\>_{AB},|\Phi\>_{AB})=\frac{1}{2^N}\tr(O_A)=\sqrt{\frac{N!}{2^{N-2}(N+2)}}\sum_{j=j_{\min}}^{N/2}\frac{(2j+1)\operatorname{sin}\frac{\pi(2j+1)}{N+2}}{\sqrt{(\frac{N}{2}-j)!(\frac{N}{2}+j+1)!}},
\end{equation}
where $j_{\min}=0 (1/2)$ when $N$ is even (odd).

\bibliographystyle{plainnat}
\bibliography{biblio2}

\begin{thebibliography}{35}
\providecommand{\natexlab}[1]{#1}
\providecommand{\url}[1]{\texttt{#1}}
\expandafter\ifx\csname urlstyle\endcsname\relax
  \providecommand{\doi}[1]{doi: #1}\else
  \providecommand{\doi}{doi: \begingroup \urlstyle{rm}\Url}\fi

\bibitem[Pan()]{PanKoperRecycling}
Git repository.
\newblock
  \url{https://bitbucket.org/pan_koper/entanglementrecycling/src/master/}.

\bibitem[ful()]{fulton_harris}
\emph{Representation {Theory} - {A} {First} {Course} {\textbar} {William}
  {Fulton} {\textbar} {Springer}}.

\bibitem[Banchi et~al.(2020)Banchi, Pereira, Lloyd, and Pirandola]{Banchi2020}
Leonardo Banchi, Jason Pereira, Seth Lloyd, and Stefano Pirandola.
\newblock Convex optimization of programmable quantum computers.
\newblock \emph{npj Quantum Information}, 6\penalty0 (1):\penalty0 42, May
  2020.
\newblock ISSN 2056-6387.
\newblock \doi{10.1038/s41534-020-0268-2}.

\bibitem[Beigi and K{\"o}nig(2011)]{beigi_konig}
Salman Beigi and Robert K{\"o}nig.
\newblock Simplified instantaneous non-local quantum computation with
  applications to position-based cryptography.
\newblock \emph{New Journal of Physics}, 13\penalty0 (9):\penalty0 093036,
  2011.
\newblock ISSN 1367-2630.
\newblock \doi{10.1088/1367-2630/13/9/093036}.

\bibitem[Bennett et~al.(1993)Bennett, Brassard, Cr{\'e}peau, Jozsa, Peres, and
  Wootters]{bennett_teleporting_1993}
Charles~H. Bennett, Gilles Brassard, Claude Cr{\'e}peau, Richard Jozsa, Asher
  Peres, and William~K. Wootters.
\newblock Teleporting an unknown quantum state via dual classical and
  {Einstein}-{Podolsky}-{Rosen} channels.
\newblock \emph{Physical Review Letters}, 70\penalty0 (13):\penalty0
  1895--1899, March 1993.
\newblock \doi{10.1103/PhysRevLett.70.1895}.

\bibitem[Boschi et~al.(1998)Boschi, Branca, De~Martini, Hardy, and
  Popescu]{boschi_experimental_1998}
D.~Boschi, S.~Branca, F.~De~Martini, L.~Hardy, and S.~Popescu.
\newblock Experimental {Realization} of {Teleporting} an {Unknown} {Pure}
  {Quantum} {State} via {Dual} {Classical} and {Einstein}-{Podolsky}-{Rosen}
  {Channels}.
\newblock \emph{Physical Review Letters}, 80\penalty0 (6):\penalty0 1121--1125,
  February 1998.
\newblock \doi{10.1103/PhysRevLett.80.1121}.

\bibitem[Buhrman et~al.(2016)Buhrman, Czekaj, Grudka, Horodecki, Horodecki,
  Markiewicz, Speelman, and Strelchuk]{buhrman_quantum_2016}
Harry Buhrman, {\L}ukasz Czekaj, Andrzej Grudka, Micha{\l} Horodecki, Pawe{\l}
  Horodecki, Marcin Markiewicz, Florian Speelman, and Sergii Strelchuk.
\newblock Quantum communication complexity advantage implies violation of a
  {Bell} inequality.
\newblock \emph{Proceedings of the National Academy of Sciences}, 113\penalty0
  (12):\penalty0 3191--3196, March 2016.
\newblock ISSN 0027-8424, 1091-6490.
\newblock \doi{10.1073/pnas.1507647113}.

\bibitem[{Chiribella} and {Ebler}(2019)]{Ebler}
Giulio {Chiribella} and Daniel {Ebler}.
\newblock {Quantum speedup in the identification of cause-effect relations}.
\newblock \emph{Nature Communications}, 10:\penalty0 1472, Apr 2019.
\newblock \doi{10.1038/s41467-019-09383-8}.

\bibitem[Christandl et~al.(2020)Christandl, Leditzky, Majenz, Smith, Speelman,
  and Walter]{majenz2}
Matthias Christandl, Felix Leditzky, Christian Majenz, Graeme Smith, Florian
  Speelman, and Michael Walter.
\newblock Asymptotic performance of port-based teleportation.
\newblock \emph{Communications in Mathematical Physics}, Nov 2020.
\newblock ISSN 1432-0916.
\newblock \doi{10.1007/s00220-020-03884-0}.

\bibitem[Fulton(1997)]{FultonSchur}
William Fulton.
\newblock \emph{Young tableaux: with applications to representation theory and
  geometry}.
\newblock Cambridge University Press, 1997.

\bibitem[Gottesman and Chuang(1999)]{gottesman_demonstrating_1999}
Daniel Gottesman and Isaac~L. Chuang.
\newblock Demonstrating the viability of universal quantum computation using
  teleportation and single-qubit operations.
\newblock \emph{Nature}, 402\penalty0 (6760):\penalty0 390--393, November 1999.
\newblock ISSN 0028-0836.
\newblock \doi{10.1038/46503}.

\bibitem[Gross and Eisert(2007)]{gross_novel_2007}
D.~Gross and J.~Eisert.
\newblock Novel {Schemes} for {Measurement}-{Based} {Quantum} {Computation}.
\newblock \emph{Physical Review Letters}, 98\penalty0 (22):\penalty0 220503,
  May 2007.
\newblock \doi{10.1103/PhysRevLett.98.220503}.

\bibitem[Ishizaka and Hiroshima(2008)]{ishizaka_asymptotic_2008}
Satoshi Ishizaka and Tohya Hiroshima.
\newblock Asymptotic {Teleportation} {Scheme} as a {Universal} {Programmable}
  {Quantum} {Processor}.
\newblock \emph{Physical Review Letters}, 101\penalty0 (24):\penalty0 240501,
  December 2008.
\newblock \doi{10.1103/PhysRevLett.101.240501}.

\bibitem[Ishizaka and Hiroshima(2009)]{ishizaka_quantum_2009}
Satoshi Ishizaka and Tohya Hiroshima.
\newblock Quantum teleportation scheme by selecting one of multiple output
  ports.
\newblock \emph{Physical Review A}, 79\penalty0 (4):\penalty0 042306, April
  2009.
\newblock \doi{10.1103/PhysRevA.79.042306}.

\bibitem[Jeong et~al.(2020)Jeong, Kim, and Lee]{jeong2020generalization}
Kabgyun Jeong, Jaewan Kim, and Soojoon Lee.
\newblock Generalization of port-based teleportation and controlled
  teleportation capability.
\newblock \emph{Phys. Rev. A}, 102:\penalty0 012414, Jul 2020.
\newblock \doi{10.1103/PhysRevA.102.012414}.

\bibitem[Jozsa(2005)]{jozsa_introduction_2005}
Richard Jozsa.
\newblock An introduction to measurement based quantum computation.
\newblock \emph{arXiv:quant-ph/0508124}, August 2005.
\newblock arXiv: quant-ph/0508124.

\bibitem[Kopszak et~al.(2021)Kopszak, Mozrzymas, Studzi{\'{n}}ski, and
  Horodecki]{Kopszak2020}
Piotr Kopszak, Marek Mozrzymas, Micha{\l{}} Studzi{\'{n}}ski, and Micha{\l{}}
  Horodecki.
\newblock Multiport based teleportation – transmission of a large amount of
  quantum information.
\newblock \emph{{Quantum}}, 5:\penalty0 576, November 2021.
\newblock ISSN 2521-327X.
\newblock \doi{10.22331/q-2021-11-11-576}.

\bibitem[Leditzky(2020)]{leditzky2020optimality}
Felix Leditzky.
\newblock {Optimality of the pretty good measurement for port-based
  teleportation}.
\newblock \emph{https://arxiv.org/abs/2008.11194}, art. arXiv:2008.11194, 2020.

\bibitem[Losonczi(1992)]{Losonczi1992}
L.~Losonczi.
\newblock Eigenvalues and eigenvectors of some tridiagonal matrices.
\newblock \emph{Acta Mathematica Hungarica}, 60\penalty0 (3):\penalty0
  309--322, Sep 1992.
\newblock ISSN 1588-2632.
\newblock \doi{10.1007/BF00051649}.

\bibitem[Mozrzymas et~al.(2014)Mozrzymas, Horodecki, and Studzi{\'n}ski]{Moz1}
Marek Mozrzymas, Micha{\l} Horodecki, and Micha{\l} Studzi{\'n}ski.
\newblock Structure and properties of the algebra of partially transposed
  permutation operators.
\newblock \emph{Journal of Mathematical Physics}, 55\penalty0 (3):\penalty0
  032202, March 2014.
\newblock ISSN 0022-2488, 1089-7658.
\newblock \doi{10.1063/1.4869027}.

\bibitem[{Mozrzymas} et~al.(2018{\natexlab{a}}){Mozrzymas}, {Studzi{\'n}ski},
  and {Horodecki}]{MozJPA}
Marek {Mozrzymas}, Micha{\l} {Studzi{\'n}ski}, and Micha{\l} {Horodecki}.
\newblock {A simplified formalism of the algebra of partially transposed
  permutation operators with applications}.
\newblock \emph{Journal of Physics A Mathematical General}, 51\penalty0
  (12):\penalty0 125202, Mar 2018{\natexlab{a}}.
\newblock \doi{10.1088/1751-8121/aaad15}.

\bibitem[{Mozrzymas} et~al.(2018{\natexlab{b}}){Mozrzymas}, {Studzi{\'n}ski},
  {Strelchuk}, and {Horodecki}]{StuNJP}
Marek {Mozrzymas}, Micha{\l} {Studzi{\'n}ski}, Sergii {Strelchuk}, and
  Micha{\l} {Horodecki}.
\newblock {Optimal port-based teleportation}.
\newblock \emph{New Journal of Physics}, 20\penalty0 (5):\penalty0 053006, May
  2018{\natexlab{b}}.
\newblock \doi{10.1088/1367-2630/aab8e7}.

\bibitem[Mozrzymas et~al.(2021)Mozrzymas, Studzi{\'{n}}ski, and
  Kopszak]{Mozrzymas2021optimalmultiport}
Marek Mozrzymas, Micha{\l{}} Studzi{\'{n}}ski, and Piotr Kopszak.
\newblock Optimal {M}ulti-port-based {T}eleportation {S}chemes.
\newblock \emph{{Quantum}}, 5:\penalty0 477, June 2021.
\newblock ISSN 2521-327X.
\newblock \doi{10.22331/q-2021-06-17-477}.

\bibitem[Murao et~al.(1999)Murao, Jonathan, Plenio, and
  Vedral]{PhysRevA.59.156}
M.~Murao, D.~Jonathan, M.~B. Plenio, and V.~Vedral.
\newblock Quantum telecloning and multiparticle entanglement.
\newblock \emph{Phys. Rev. A}, 59:\penalty0 156--161, Jan 1999.
\newblock \doi{10.1103/PhysRevA.59.156}.

\bibitem[Nielsen and Chuang(1997)]{Nielsen1997}
M.~A. Nielsen and Isaac~L. Chuang.
\newblock Programmable quantum gate arrays.
\newblock \emph{Phys. Rev. Lett.}, 79:\penalty0 321--324, Jul 1997.
\newblock \doi{10.1103/PhysRevLett.79.321}.

\bibitem[Pereira et~al.(2021)Pereira, Banchi, and Pirandola]{sim}
Jason Pereira, Leonardo Banchi, and Stefano Pirandola.
\newblock Characterising port-based teleportation as universal simulator of
  qubit channels.
\newblock \emph{Journal of Physics A: Mathematical and Theoretical},
  54\penalty0 (20):\penalty0 205301, apr 2021.
\newblock \doi{10.1088/1751-8121/abe67a}.

\bibitem[Pirandola et~al.(2015)Pirandola, Eisert, Weedbrook, Furusawa, and
  Braunstein]{pirandola_advances_2015}
S.~Pirandola, J.~Eisert, C.~Weedbrook, A.~Furusawa, and S.~L. Braunstein.
\newblock Advances in quantum teleportation.
\newblock \emph{Nature Photonics}, 9\penalty0 (10):\penalty0 641--652, October
  2015.
\newblock ISSN 1749-4885.
\newblock \doi{10.1038/nphoton.2015.154}.

\bibitem[{Pirandola} et~al.(2019){Pirandola}, {Laurenza}, {Lupo}, and
  {Pereira}]{limit}
Stefano {Pirandola}, Riccardo {Laurenza}, Cosmo {Lupo}, and Jason~L. {Pereira}.
\newblock {Fundamental limits to quantum channel discrimination}.
\newblock \emph{npj Quantum Information}, 5:\penalty0 50, Jun 2019.
\newblock \doi{10.1038/s41534-019-0162-y}.

\bibitem[Quintino et~al.(2019)Quintino, Dong, Shimbo, Soeda, and
  Murao]{PhysRevLett.123.210502}
Marco~T\'ulio Quintino, Qingxiuxiong Dong, Atsushi Shimbo, Akihito Soeda, and
  Mio Murao.
\newblock Reversing unknown quantum transformations: Universal quantum circuit
  for inverting general unitary operations.
\newblock \emph{Phys. Rev. Lett.}, 123:\penalty0 210502, Nov 2019.
\newblock \doi{10.1103/PhysRevLett.123.210502}.

\bibitem[Raussendorf and Briegel(2001)]{raussendorf_one-way_2001}
Robert Raussendorf and Hans~J. Briegel.
\newblock A {One}-{Way} {Quantum} {Computer}.
\newblock \emph{Physical Review Letters}, 86\penalty0 (22):\penalty0
  5188--5191, May 2001.
\newblock \doi{10.1103/PhysRevLett.86.5188}.

\bibitem[{Sedl{\'a}k} et~al.(2019){Sedl{\'a}k}, {Bisio}, and {Ziman}]{Stroing}
Michal {Sedl{\'a}k}, Alessandro {Bisio}, and M{\'a}rio {Ziman}.
\newblock {Optimal Probabilistic Storage and Retrieval of Unitary Channels}.
\newblock \emph{\prl}, 122\penalty0 (17):\penalty0 170502, May 2019.
\newblock \doi{10.1103/PhysRevLett.122.170502}.

\bibitem[Stein et~al.(2022)]{sage}
W.\thinspace{}A. Stein et~al.
\newblock \emph{{S}age {M}athematics {S}oftware ({V}ersion 9.5)}.
\newblock The Sage Development Team, 2022.
\newblock {\tt http://www.sagemath.org}.

\bibitem[Strelchuk et~al.(2013)Strelchuk, Horodecki, and
  Oppenheim]{strelchuk_generalized_2013}
Sergii Strelchuk, Micha{\l} Horodecki, and Jonathan Oppenheim.
\newblock Generalized {Teleportation} and {Entanglement} {Recycling}.
\newblock \emph{Physical Review Letters}, 110\penalty0 (1):\penalty0 010505,
  January 2013.
\newblock \doi{10.1103/PhysRevLett.110.010505}.

\bibitem[{Studzi{\'n}ski} et~al.(2017){Studzi{\'n}ski}, {Strelchuk},
  {Mozrzymas}, and {Horodecki}]{Studzinski2017}
Micha{\l} {Studzi{\'n}ski}, Sergii {Strelchuk}, Marek {Mozrzymas}, and
  Micha{\l} {Horodecki}.
\newblock {Port-based teleportation in arbitrary dimension}.
\newblock \emph{Scientific Reports}, 7:\penalty0 10871, Sep 2017.
\newblock \doi{10.1038/s41598-017-10051-4}.

\bibitem[Studziński et~al.(2020)Studziński, Mozrzymas, Kopszak, and
  Horodecki]{stud2020A}
Michał Studziński, Marek Mozrzymas, Piotr Kopszak, and Michał Horodecki.
\newblock Efficient multi-port teleportation schemes.
\newblock \emph{https://arxiv.org/abs/2008.00984}, 2020.

\end{thebibliography}
\end{document}